\documentclass[11pt]{article}

\usepackage[dvipsnames,svgnames,table]{xcolor}
\usepackage{enumitem}
\usepackage{graphicx}
\usepackage{caption}
\usepackage{subcaption}
\usepackage{fancybox}
\usepackage{comment}
\usepackage{nameref}
\usepackage{mathtools}

\usepackage{tikz}
\usetikzlibrary{bbox}

\usepackage{amsmath}
\usepackage{amssymb}
\usepackage{amsthm}
\usepackage[bottom]{footmisc}

\usepackage[ruled, noend, linesnumbered]{algorithm2e}

\usepackage{thm-restate}

\usepackage{float}
\definecolor{ForestGreen}{rgb}{0.1333,0.5451,0.1333}
\definecolor{DarkRed}{rgb}{0.8,0,0}
\definecolor{Red}{rgb}{1,0,0}
\usepackage[linktocpage=true,
pagebackref=true,colorlinks,
linkcolor=ForestGreen,citecolor=ForestGreen,
bookmarks,bookmarksopen,bookmarksnumbered]
{hyperref}

\usepackage{cleveref}
\AddToHook{cmd/appendix/before}{\crefalias{section}{appendix}} 

\newcommand{\authornewline}{%
\end{tabular}\par\medskip
\begin{tabular}[t]{c}%
}
\newcommand{\authorblock}[3]{%
\parbox[t]{0.4\textwidth}{%
\centering
#1\par
\small #2\par
\small \texttt{#3}\par
}%
}

\newtheorem{theorem}{Theorem}[section]

\newtheorem{lemma}[theorem]{Lemma}
\newtheorem{observation}[theorem]{Observation}

\newtheorem{claim}[theorem]{Claim}

\newtheorem{definition}[theorem]{Definition}

\newtheorem*{theorem*}{Theorem}
\newtheorem*{corollary*}{Corollary}
\newtheorem*{conjecture*}{Conjecture}
\newtheorem*{lemma*}{Lemma}
\newtheorem*{thm*}{Theorem}
\newtheorem*{prop*}{Proposition}
\newtheorem*{obs*}{Observation}
\newtheorem*{definition*}{Definition}

\newtheorem*{remark*}{Remark}
\newtheorem*{rec*}{Recommendation}

\newenvironment{fminipage}%
  {\begin{Sbox}\begin{minipage}}%
  {\end{minipage}\end{Sbox}\fbox{\TheSbox}}

\newcommand\Otil{\widetilde{O}}

\newcommand{\E}[1]{\mathop{{}\mathbb{E}}\left[#1\right]}

\renewcommand{\E}{\mathbb{E}}

\DeclareMathOperator{\dist}{dist}
\DeclareMathOperator{\diam}{diam}

\global\long\def\Otil{\tilde{O}}

\newcommand{\cut}{\operatorname{cut}}

\SetKwComment{Comment}{\textcolor{gray}{$\triangleright$\ }}{}
\SetKwRepeat{Do}{do}{while}

\usepackage{fullpage}

\title{Low-Stretch Spanning Trees\\ via Smoothed Analysis of Dijkstra's Algorithm}
\author{
  \authorblock{Ioannis Dorkofikis\thanks{The research leading to these results has received funding from the starting grant “A New Paradigm for Flow and Cut Algorithms” (no. TMSGI2 218022) and grant no. 200021 204787 of the Swiss National Science Foundation.}}%
  {ETH Zürich}%
  {ioannis.dorkofikis@inf.ethz.ch}
  \and
  \authorblock{Bernhard Haeupler\thanks{The research leading to these results was partially funded by the Ministry of Education and Science of Bulgaria's support for INSAIT as part of the Bulgarian National Roadmap for Research Infrastructure and by the European Research Council (ERC) under the European Union's Horizon Europe research and innovation program (ERC Advanced grant agreement 101268062 and ERC Starting grant agreement 949272).}}%
  {INSAIT, Sofia University ``St.~Kliment Ohridski'' and ETH Zürich}%
  {bernhard.haeupler@insait.ai}
  \authornewline
  \authorblock{Maximilian Probst Gutenberg\footnotemark[1]}%
  {ETH Zürich}%
  {maximilian.probst@inf.ethz.ch}
  \and
  \authorblock{Antti Roeyskoe\footnotemark[2]}%
  {ETH Zürich}%
  {antti.roeyskoe@inf.ethz.ch}
  \authornewline
  \authorblock{Aurelio Sulser\footnotemark[1]}%
  {ETH Zürich}
  {aurelio.sulser@inf.ethz.ch}
  \and
  \authorblock{Gernot Z\"ocklein\footnotemark[1]}%
  {ETH Zürich}%
  {gernot.zoecklein@inf.ethz.ch}
}
\date{}

\date{}

\begin{document}
\maketitle
\begin{abstract}

Given an undirected weighted graph $G$, a $\gamma$-approximate low-stretch spanning tree (LSST) $T \subseteq G$ is a tree that approximates the distance metric of $G$ up to a $\gamma$-factor in expectation. Currently, existing algorithms to find a provably good LSST carefully construct an approximate shortest-path tree from an arbitrary source. The resulting algorithms are intricate. In contrast, practitioners observed that a much simpler heuristic performs surprisingly well: choose an arbitrary root, run Dijkstra's algorithm, and use the resulting shortest-path tree as an LSST. 

In this paper, we give a smoothed analysis of shortest-path tree algorithms, such as Dijkstra's algorithm, that explains this behavior. We show that adding a small perturbation to the weights of the input graph suffices to turn the shortest path tree rooted at an arbitrary node in the resulting graph into an $\tilde{O}(1)$-approximate LSST.

We further show that the set of perturbations can be computed efficiently from few low-diameter decompositions (LDDs). Thus, our proof is also constructive in the sense of giving a novel approach to computing LSSTs.
\end{abstract}

\pagenumbering{gobble}

\pagebreak

\pagenumbering{arabic}

\section{Introduction}

Low-stretch spanning trees, and more broadly tree embeddings, are the cornerstone of a now-standard paradigm in graph algorithms: replace the input graph by a single tree that approximates its metric, on which cuts, flows, distances, and linear-algebraic primitives become more tractable and easier to reason about. Hence, low-stretch spanning trees and tree embeddings, as defined below, have been studied intensively over the past decades \cite{alon1995graph, bartal1996probabilistic, bartal1998approximating, fakcharoenphol2003tight, bartal2004graph, elkin2005lower, abraham2008nearly, abraham2012using, forster2019dynamic, chechik2020dynamic, forster2021dynamic, rozhovn2022deterministic, becker2024decentralized, kyng2025random, fletcher2026spanning}.

\begin{definition}[Low-Stretch Spanning Tree]
Given a weighted undirected graph $G = (V,E,w:E\rightarrow [W])$, we say that a tree $T \subseteq G$ drawn from some distribution $\mathcal{T}$ over spanning trees of $G$ is a \emph{randomized $\gamma$-approximate low-stretch spanning tree} if for every vertex pair $u,v \in V$,
\begin{equation*}
\E_{T \sim \mathcal{T}}[\dist_T(u,v)] \le \gamma \dist_G(u,v).    
\end{equation*}
A tree embedding is defined analogously but without the restriction that the trees $T$ have to be subgraphs of $G$.
\end{definition}
 
 Most notably, this paradigm underlies the first near-linear time solvers for Laplacian systems \cite{spielman2004nearly, koutis2014approaching, koutis2011nearly, kelner2013simple}, where low-stretch spanning trees serve as preconditioners. More recently, it has also been a key ingredient in the first provably fast almost-linear-time algorithms for min-cost flow \cite{chen2025maximum, van2023deterministic, chen2024almost}.

\paragraph{Low-Stretch Spanning Trees in Practice.} In practice, LSSTs were first implemented and evaluated \cite{hoske2015nearly, deweese2016empirical} 
as a primitive for practical Laplacian graph solvers. Very recently, they have been used in the implementation \cite{mincostFlowImpl} of the min-cost flow solver from \cite{chen2025maximum}. During their implementation, the authors of \cite{mincostFlowImpl} made a curious observation. They noted that the often intricate algorithms for computing low-stretch trees are matched, and sometimes even outperformed, by a strikingly simple heuristic: run Dijkstra's algorithm from an arbitrary vertex and use the resulting shortest-path tree as a low-stretch spanning tree. 
 
From the perspective of algorithm analysis, this observation is somewhat matched by the fact that existing polylogarithmic-stretch constructions can be viewed, at a high level, as carefully engineered approximate shortest-path trees grown from an arbitrary source (c.f. \cite{elkin2005lower, abraham2008nearly, abraham2012using, fletcher2026spanning}). This gives an intuitive explanation for why such a phenomenon seems plausible in practice. On the other hand, simple examples show that the \emph{exact} shortest-path tree from an arbitrary source can have very poor stretch. 

\paragraph{Beyond Worst-Case Analysis.} Since worst-case analysis cannot fully explain the behavior seen in practice, the algorithm analysis community has introduced average-case analysis and smoothed analysis \cite{spielman2004smoothed}. While the former framework often lends an interesting perspective, as noted in \cite{spielman2004smoothed}, it is often unconvincing to analyze performance over a random distribution of instances, as real-world instances are often highly structured and might thus look very different from the average input. 

Instead \cite{spielman2004smoothed} introduced smoothed analysis, which analyzes algorithmic performance under small perturbations to the input. More precisely, given the set of input instances $I \in \mathcal{I}$ and an appropriate definition of a small neighborhood $N(I)$ for each instance, along with a distribution over the neighborhood,\footnote{Note that formally, our perturbations can be large, but most mass is on a small neighborhood, and mass is decaying with increasing distance. This is similar to the setup in \cite{spielman2004smoothed}, which adds Gaussian perturbations.} it measures the expected cost of algorithm $\mathcal{A}$ evaluated on the small neighborhood of each instance
\begin{equation*}
    \max_{I \in \mathcal{I}} \mathbb{E}_{I' \sim N(I)}[\textnormal{cost}(\mathcal{A}(I'))].
\end{equation*}

In this paper, we ask the following question: 
\begin{center}    
\emph{Given any input graph $G$, can we apply a small perturbation to its weights so that thereafter any exact shortest-path tree from an arbitrary vertex is provably a good low-stretch spanning tree?}
\end{center}

Equivalently, we ask whether a smoothed analysis~\cite{spielman2004smoothed} of Dijkstra's algorithm for computing low-stretch spanning trees can recover the behavior observed in practice. 
 
\subsection{Our Contribution} 

We answer the above question in the affirmative: a simple perturbation of the edge weights suffices to recover a polylogarithmic stretch guarantee.

\begin{restatable}{theorem}{mainthm}\label{thm:main}
Let $G=(V,E,w)$ be a graph with integral, polynomially-bounded edge weights, and let $\varepsilon \in [\frac{1}{n}, 1]$. Then, there is a random perturbation $\Delta \ge 0$ satisfying $\mathbb{E}[\Delta(e)] = \varepsilon \cdot O(w(e))$ for each $e \in E$, such that, for any $u,v,r \in V$, the shortest-path tree $\widetilde{T}_r$ rooted at $r$ in $\widetilde{G} = (V, E, \widetilde{w} = w + \Delta)$ satisfies that
\begin{equation*}
\E\left[ \frac{\dist_{\widetilde{T}_r}(u, v)}{\dist_{\widetilde{G}}(u, v)}\right] = \Otil\left(\frac{1}{\varepsilon}\right).    
\end{equation*}
\end{restatable}

We remark that the inverse dependence of the algorithm performance on the size of the noise $\varepsilon$ is intuitively to be expected and standard in smoothed analysis. One way to interpret \Cref{thm:main} is that, viewing weighted graphs as points in $\mathbb{R}^E$, every graph is ``close'' to many graphs on which the SSSP tree is already well-behaved. We point out that while each edge weight is roughly preserved in expectation, the chosen perturbation leads to highly correlated weights. As discussed later, this is necessary to achieve provable guarantees.

The SSSP tree $\widetilde{T}$ of the randomly perturbed graph $\widetilde{G}$ additionally has the desirable property of also being a randomized low-stretch spanning tree of the \emph{original} graph $G$. As our choice of perturbations can be constructed rather straightforwardly from few low-diameter decompositions, this yields a second perspective on our result: our analysis provides an efficient, simple, and completely novel algorithm to compute a low-stretch spanning tree with small polylogarithmic loss.\footnote{We believe this may even be the more interesting statement. Thus, our result may also be interpreted to be somewhat closer in spirit to \cite{kelner2006randomized} where the solution to a randomly perturbed instance is efficiently mapped back to a solution of the original instance.}

\begin{restatable}{theorem}{LowStretchTree}
\label{thm:LowStretchTree}
Given an input graph $G = (V,E,w)$, there is a simple $\widetilde{O}(m)$-time algorithm based on random perturbations to compute a randomized $O(\log^5 n)$-approximate low-stretch spanning tree.
\end{restatable}


\subsection{Technical Overview} 

We next give a high-level overview, some intuition for the design of the perturbations used in \Cref{thm:LowStretchTree}, and then an overview of our proof. We work with a connected, simple input graph $G=(V,E,w)$ with positive, polynomially bounded, integer edge weights $w : E \rightarrow [W]$.

\paragraph{A First Attempt at Perturbing the Graph.} We first describe the most natural perturbation strategy and outline the difficulties that we encounter in the process. We will require $\mathbb{E}[\Delta(e)] = O(w(e))$ (otherwise the perturbation has more influence than the original metric) and $\Delta \geq 0$ since the shortest path tree is undefined in the presence of negative edges. 

The most natural approach is to consider i.i.d. sampling of the perturbation.  However, we next demonstrate that sampling perturbations i.i.d. fails on the cycle graph $C_n$. Recall first, that on the cycle graph, any LSST is formed by removing a single edge $e$ from $C_n$, which yields stretch $n-1$ for the removed edge. Consider the shortest path tree $T_r$ from some arbitrary vertex $r$. On $C_n$, it omits the edge opposite to the root $r$. In order to recover $\widetilde{O}(1)$ stretch for every edge, we need to redistribute the omitted edge over at least $\widetilde{\Omega}(n)$ edges with good probability. 

Now consider the case where perturbations are i.i.d. and bounded, i.e. each $\Delta(e) \in [0, C]$ for some universal constant $C > 0$. In this case, we can apply a Chernoff bound to get concentration of the distances in the perturbed graph, which yields that no vertex is at distance further than $O(\sqrt{n})$ from expected w.h.p. But this distributes the omitted edge only over $\widetilde{O}(\sqrt{n})$ edges opposite to $r$. Thus, on average, these edges are still left with $\widetilde{\Omega}(\sqrt{n})$ stretch.

Let us therefore assume that random variables $\Delta(e)$ are unbounded. However, consider now the graph where we replace each edge on $C_n$ with $\alpha = 10 \log n$ multi-edges (again of weight $1$). Distance computations only consider the edge among the multi-edges with the smallest perturbation. But by Markov, any fixed edge $e$ has perturbation at most $2 \cdot \mathbb{E}[\Delta(e)] = O(w(e))$ with probability at least $1/2$. Consequently, one of the multi-edges thus has weight $O(w(e))$ w.h.p., and this is true for all original edges in $C_n$. Since we can condition on this event, we thus find ourselves again in the previous bad scenario.

We conclude that i.i.d. sampling does not seem sensible to obtain the desired perturbations.

\paragraph{Correlating Perturbations via Low-Diameter Decompositions.} Reflecting on the above arguments, it seems essential that parallel edges are strongly correlated. More broadly, this could be interpreted as a notion where perturbations of edges should be more correlated, the closer they are. 

A natural way to relate nearby edges is low-diameter decompositions. 

\begin{definition}[Low Diameter Decomposition]\label{def:ldd}
Given $G = (V,E,w)$ and a diameter parameter $D > 0$, a \emph{low-diameter decomposition (LDD)} with slack $s > 0$ is a randomized partition $\mathcal{C}$ of the vertex set $V$ that satisfies the following two properties:
\begin{enumerate}
    \item \textbf{Diameter.} Every cluster $C \in \mathcal{C}$ has strong diameter $\diam(G[C])$ at most $D$ with high probability.
    
    \item \textbf{Separation probability.} Any two vertices $u,v \in V$ are separated (appear in distinct clusters) with probability at most $s \cdot
    \dfrac{\dist_G(u,v)}{D}$. \label{prop:sepprob}
\end{enumerate}
\end{definition}

There are multiple known algorithms to compute LDDs with optimal slack $s=O(\log n)$; see, for example, \cite{bartal1996probabilistic, CalinescuKR00ckroriginal, miller2013mpxoriginal}. In this work, we consider geometric ball carving \cite{bartal1996probabilistic} and MPX \cite{miller2013mpxoriginal}. We describe these algorithms later.

We call the partition $\mathcal{C}$ a clustering. For a cluster $C \in \mathcal{C}$, we write $\delta(C)=\delta_G(C)$ to denote all the edges of $G$ with exactly one endpoint in $C$, and write $\delta(\mathcal{C}) \coloneqq \bigcup_{C \in \mathcal{C}} \delta(C)$. We say an edge $e = (u, v) \in E$ is \textit{cut} by $\mathcal{C}$ if $e \in \delta(\mathcal{C})$, or equivalently if $u$ and $v$ appear in distinct clusters of $\mathcal{C}$.

\paragraph{A Formal Perturbation Algorithm.} For a growth factor $\alpha = \Theta(\log n)$, let $L \coloneqq \lceil \log_\alpha nW \rceil = O(\log n)$ so that $\diam(G) \leq \alpha^L$. Our algorithm proceeds in three steps:
\begin{enumerate}
    \item We first compute a hierarchy of low-diameter decompositions $\mathcal{C}_L, \dots, \mathcal{C}_0$ where $\mathcal{C}_L \coloneqq \{V\}$, each $\mathcal{C}_{\ell}$ for $\ell \in \{L - 1, \dots, 1\}$ is a low-diameter decomposition with diameter parameter $D_\ell = \alpha^\ell$ of the clusters of $\mathcal{C}_{\ell + 1}$ (i.e. of the graph $G \setminus \delta(\mathcal{C}_{\ell + 1})$), and $\mathcal{C}_0 \coloneqq \{\{v\} \mid v \in V\}$.
    \item Then, for any edge $e$ let $\ell_e \coloneqq \max \{\ell : e \in \delta(\mathcal{C}_\ell)\}$ be the highest level such that $e$ is separated by $\mathcal{C}_\ell$. For each edge $e$, sample a perturbation
    \begin{equation*}
    \Delta(e) = \epsilon D_{\ell_e} - X_e,
    \end{equation*}
    where $\epsilon = O(1 / \log^2 n)$ is a noise magnitude parameter and each $X_e$ is an independent geometric random variable with mean $\Theta\left(\epsilon D_{\ell_e} / \log n\right)$.\footnote{This mean is chosen so that $X_e \leq \epsilon D_{\ell_e}$ with high probability. Technically, if $\Delta(e) < 0$ occurs, we set $\Delta(e) = 0$.}
    
    We denote the perturbed graph with edge weights $\widetilde{w}(e) \coloneqq w(e) + \Delta(e)$ by $\widetilde{G}$.

    \item Finally, for an arbitrary source $r$ (chosen independently of the realization of $\Delta$), output the shortest path tree $T_r$ in $\widetilde{G}$, breaking ties arbitrarily. 
\end{enumerate}
We write $\widetilde{T}$ for the resulting tree $T$ with edge weights inherited from the perturbed graph $\widetilde{G}$.

\paragraph{Sketch of the Analysis.} 

We now outline the proof of \Cref{thm:LowStretchTree} by showing that there exist LDD algorithms for which this perturbation algorithm is guaranteed to output an $\widetilde{O}(1)$-approximate LSST. 
For this, we fix $\epsilon = \Theta(1 / \log^2 n)$. Our full analysis extends to smaller $\epsilon$ with a linear loss in the obtained stretch guarantee, which obtains \Cref{thm:main}.





By linearity of expectation, it suffices to bound the expected stretch of every edge:
\begin{equation*}
\frac{\mathbb{E}[\dist_{\widetilde{T}}(u,v)]}{w(u,v)} = O(\log^5 n), \qquad \forall (u,v) \in E.
\end{equation*}
To this end, we establish the following level-by-level bound: for some constant $c > 0$, for every $\ell \in [L]$ and $(u,v) \in E$,


\begin{equation}\label{eq:technical-main}
    \mathbb{P}\left[\dist_{\widetilde{T}}(u,v) \geq c D_\ell\right] \leq
\frac{w(u,v)}{D_\ell}\cdot O\left(\log^3 n\right).
\end{equation}

It immediately follows from \Cref{eq:technical-main} that the expected stretch of the output tree satisfies the desired bound. Indeed, summing over all $L$ possible ranges $[c D_\ell, c D_{\ell + 1})$ for $\dist_{\widetilde{T}}(u, v)$, the expected stretch contribution incurred for each is $\alpha \cdot O(\log^3 n)$, for a total of $O(\log^5 n)$.

To prove \Cref{eq:technical-main}, fix a level $\ell \in [L]$ and an edge $(u,v) \in E$. The LDD guarantees ensure that with probability at least $1-\frac{w(u,v)}{D_\ell} \cdot O(\log n)$, $u$ and $v$ are in the same cluster in $\mathcal{C}_\ell$. Fix a hierarchy $\{\mathcal{C}_\ell\}_\ell$ where this occurs, so that $u,v \in C \in \mathcal{C}_\ell$. If the diameter $\diam(\widetilde{G}[C])$ of the cluster did not get blown up by more than a constant by weight increases from levels of the hierarchy below $\ell$, then the distance between $u$ and $v$ in $\widetilde{G}$ is still $O(D_\ell)$. Of course, this is not sufficient for them to be close in $\widetilde{T}$, as their lowest common ancestor could be far away.
However, if the two vertices did have a common ancestor $x$ also in the cluster $C$, then by the triangle inequality and the fact that subpaths of shortest paths are shortest paths, we would have 
\begin{align*}    
\dist_{\widetilde{T}}(u,v)
&\leq \dist_{{\widetilde{T}}}(x,u) + \dist_{{\widetilde{T}}}(x,v)\\
&\leq \dist_{\widetilde{G}}(x,u) + \dist_{\widetilde{G}}(x,v)\\
&\leq 2 \diam(\widetilde{G}[C])\\
&= O(D_\ell).
\end{align*}
This establishes two clear objectives:
\begin{enumerate}
    \item Bounding the probability that no common ancestor appears inside $C$.
    \item Proving guarantees on the diameters of clusters in the perturbed graph.
\end{enumerate}

\paragraph{Objective 1.} We first focus on showing that a common ancestor exists in $C$ with sufficiently high probability. If $r \in C$, then $r$ can serve as the common ancestor of $u$ and $v$ inside the cluster; thus, we may assume that $r \not\in C$.
Let $\widetilde{T}(r,u)$ be the $r$--$u$ path in $\widetilde{T}$, and $\widetilde{T}(r,v)$ likewise the $r$--$v$ path. Then, $u$ and $v$ have a common ancestor in $C$ if and only if $\widetilde{T}(r,u)$ and $\widetilde{T}(r,v)$ share an edge on the boundary $\delta(C)$ of $C$, as these are tree paths and cannot rejoin after diverging.

Write $P_1$ for $\widetilde{T}(r,u)$ and $P'_2$ for $\widetilde{T}(r,v)$ extended by the edge $(v, u)$. These are now two paths from $r$ to $u$. These paths have similar lengths as long as $\widetilde{w}(u, v)$ is not overly stretched, as
\begin{equation*}
    \widetilde{w}(P_1) \leq \widetilde{w}(P'_2) = \widetilde{w}(u, v) + \widetilde{w}(\widetilde{T}(r, v)) \leq 2 \widetilde{w}(u, v) + \widetilde{w}(P_1),
\end{equation*}
thus $\widetilde{w}(P'_2) - \widetilde{w}(P_1) \leq 2 \widetilde{w}(u, v)$ always holds. Let $P_2$ be the shortest $r$--$u$ path that does not share any edges with $P_1$ on the boundary of $C$ and let $I_{\text{gap}}$ be the event that
\begin{equation*}
    \widetilde{w}(P_2) - \widetilde{w}(P_1) > 2 \widetilde{w}(u, v).
\end{equation*}
Then, if $I_{\text{gap}}$ occurs, there exists a common ancestor inside $C$, as $P_2'$ must not be a candidate for $P_2$. We want to show that this event, i.e. the presence of a large gap between the shortest and second-shortest way for $r$--$u$ paths to enter $C$, is likely.

This setting is similar to the separation probability analysis of \cite{miller2013mpxoriginal}. We adapt their argument, which relies on the fact that in the perturbations $\Delta(e) = \epsilon D_{\ell_e} - X_e$ we selected, the variables $X$ follow the \emph{memoryless} geometric distribution. 

Fix the values $X_e$ for each $e \not\in \delta(C)$. For each crossing edge $e \in \delta(C)$ we think of $X_e$ as follows. Associated with every edge $e$ is a clock that is initially alive and displays $0$. The time at which the clock breaks is geometrically distributed, and $X_e$ is defined to be the value displayed by the clock when it breaks. At each time step, we pick one alive clock and reveal one more unit of its lifetime: either the clock breaks, in which case its displayed value is fixed forever, or it survives and its displayed value increases by one. By the memorylessness of the geometric distribution, conditioned on still being alive, the remaining time until the clock breaks has the same distribution.

The key observation is that, by revealing clocks outside the current winning path whenever possible, we can maintain the invariant that the current winning path always contains an alive clock. Eventually, only one alive clock remains. This final clock belongs to the winning path but to no competing path that avoids its boundary edges. By the memorylessness of the geometric distribution, its remaining lifetime is an independent geometric random variable, which contributes only to the winning path. Consequently, $\widetilde{w}(P_2)-\widetilde{w}(P_1)$ stochastically dominates a geometric random variable with stopping probability $O(\frac{1}{\epsilon D_\ell} \log n) = O(\frac{1}{D_\ell} \log^3 n)$. In particular,
\begin{equation*}
\mathbb{P}(I_{\text{gap}}) = \mathbb{P}\left(\widetilde{w}(P_2) - \widetilde{w}(P_1) > 2 \widetilde{w}(u, v)\right) \geq 1 - \frac{\widetilde{w}(u,v)}{D_\ell}O(\log^3 n).
\end{equation*}

Finally, we address the fact that we want to bound the probability of $I_{\text{gap}}$ in terms of the original weight $w(u,v)$ rather than the perturbed weight $\widetilde w(u,v)$. Consider fixing just the levels $\mathcal{C}_L, \dots, \mathcal{C}_\ell$ of the hierarchy. If at this point $u, v \in C \in \mathcal{C}_\ell$, then $\mathbb{E}[\widetilde{w}(u, v) \mid \mathcal{C}_L, \dots, \mathcal{C}_\ell] = O(w(u, v))$ still holds. We then apply the law of total probability over the randomness of selecting $\mathcal{C}_{\ell - 1}, \dots, \mathcal{C}_1$ to obtain
\begin{equation*}
\mathbb{P}(I_{\text{gap}} \mid \mathcal{C}_L, \dots, \mathcal{C}_\ell) \geq 1 - \frac{\mathbb{E}[\widetilde{w}(u,v) \mid \mathcal{C}_L, \dots, \mathcal{C}_\ell]}{D_\ell}O(\log^3 n) = 1 - \frac{w(u, v)}{D_\ell}O(\log^3 n).
\end{equation*}
Since $u$ and $v$ are coclustered in $\mathcal{C}_\ell$ with probability $1 - \frac{w(u, v)}{D_\ell} O(\log n)$, this completes the proof.

\paragraph{Objective 2.}
In the remainder of the technical overview, we justify our assumption that the diameter guarantees of the LDD hierarchy are preserved within constants in the perturbed graph. 
We begin by showing bounds for distances under perturbations caused by a single LDD that hold \emph{with high probability}. We then provide a blackbox reduction from such high-probability \emph{distance preservation} at a single level to high-probability \emph{diameter preservation} for the entire hierarchy.


\paragraph{Distance Preserving LDDs.} 

Consider computing a single clustering $\mathcal{C} = \textsc{LDD}(G, D)$, and increasing weights of cut edges by $\epsilon D$. What high-probability bounds on distances in the resulting graph $\widetilde{G}$ can we obtain?


Recall that the probability an LDD algorithm with optimal slack $O(\log n)$ cuts an edge $e$ is at most $\frac{w(e)}{D} \cdot O(\log n)$. Thus, the weight of any fixed edge or path \emph{in expectation} increases multiplicatively by at most $(1 + O(\epsilon \log n))$. A standard Chernoff bound also shows that if all edges were cut independently,
\begin{equation}\label{eq:concentration-distance-independent}
    \widetilde{w}(P) \leq (1 + O(\epsilon \log n)) \cdot w(P) + O(\epsilon D \log n)
\end{equation}
would hold with high probability for a fixed path $P$, and thus also for distances.

The obstacle to an immediate concentration bound is the correlation between nearby edges being cut. Consider for example the geometric ball carving LDD algorithm (\Cref{alg:geom}), which repeatedly picks a center $x$ in the remaining graph and a geometrically distributed radius $R$, and then carves out a radius-$R$ ball around $x$. In a situation such as \Cref{fig:octopus-1}, there is a radius for which this ball cuts the path a large number of times. Since any fixed radius is sampled with a non-negligible probability, a concentration bound on the weight of a \emph{fixed} path in $\widetilde{G}$ is impossible. 

However, the effect of a single cluster on the \emph{distance} between two vertices can be bounded, as a path in $\widetilde{G}$ can ``shortcut'' through the cluster $C$ to only be cut twice by it: see \Cref{fig:octopus-2}.

\begin{figure}[h]
\centering

\begin{tikzpicture}[bezier bounding box]
  \coordinate (v) at (0,3);
  \coordinate (s) at (-5,0);
  \coordinate (t) at (5,0);

  \foreach \x in {-3,-1,1,3}
    \draw[thick] (v)
      .. controls ({0.5*\x},2.2) and ({0.85*\x},1.2) .. (\x,0);

  \draw[red,thick]
    (-3.65,0)
    .. controls (-3.65,{16/3}) and (3.65,{16/3}) .. (3.65,0)
    .. controls (3.65,-0.28) and (3.3575,-0.5) .. (3,-0.5)
    .. controls (2.3,-0.5) and (2.3,0.5) .. (2,0.5)
    .. controls (1.7,0.5) and (1.7,-0.5) .. (1,-0.5)
    .. controls (0.3,-0.5) and (0.3,0.5) .. (0,0.5)
    .. controls (-0.3,0.5) and (-0.3,-0.5) .. (-1,-0.5)
    .. controls (-1.7,-0.5) and (-1.7,0.5) .. (-2,0.5)
    .. controls (-2.3,0.5) and (-2.3,-0.5) .. (-3,-0.5)
    .. controls (-3.3575,-0.5) and (-3.65,-0.28) .. (-3.65,0)
    -- cycle;

  \begin{scope}[blue!75!black,line width=1.7pt,
    dash pattern=on 5pt off 3pt,dash phase=1pt]
    \draw (s) -- (-3,0);
    \draw (t) -- (3,0);
    \draw (-3,0) -- (-1,0);
    \draw (3,0) -- (1,0);
    \draw (1,0) -- (-1,0);
  \end{scope}

  \foreach \p in {s,t,v}
    \fill (\p) circle (1.6pt);
  \foreach \x in {-3,-1,1,3}
    \fill (\x,0) circle (1.6pt);
  \node[above] at (v) {$x$};
  \node[left] at (s) {$s$};
  \node[right] at (t) {$t$};
  \node[blue!75!black,above] at (-4.4,0) {$P$};
  \node[red,right] at (3.4,2) {$B(x, R)$};
  
\end{tikzpicture}

\caption{A ball around $x$ with an appropriate radius cuts the $s$-$t$ path $P$ a large number of times.}
\label{fig:octopus-1}

\vspace{0.4cm}

\begin{tikzpicture}[bezier bounding box]
  \coordinate (v) at (0,3);
  \coordinate (s) at (-5,0);
  \coordinate (t) at (5,0);

  \draw[thick] (s) -- (t);
  \foreach \x in {-3,-1,1,3}
    \draw[thick] (v)
      .. controls ({0.5*\x},2.2) and ({0.85*\x},1.2) .. (\x,0);

  \draw[red,thick]
    (-3.65,0)
    .. controls (-3.65,{16/3}) and (3.65,{16/3}) .. (3.65,0)
    .. controls (3.65,-0.28) and (3.3575,-0.5) .. (3,-0.5)
    .. controls (2.3,-0.5) and (2.3,0.5) .. (2,0.5)
    .. controls (1.7,0.5) and (1.7,-0.5) .. (1,-0.5)
    .. controls (0.3,-0.5) and (0.3,0.5) .. (0,0.5)
    .. controls (-0.3,0.5) and (-0.3,-0.5) .. (-1,-0.5)
    .. controls (-1.7,-0.5) and (-1.7,0.5) .. (-2,0.5)
    .. controls (-2.3,0.5) and (-2.3,-0.5) .. (-3,-0.5)
    .. controls (-3.3575,-0.5) and (-3.65,-0.28) .. (-3.65,0)
    -- cycle;

  \draw[white,line width=2.8pt]
    (s) -- (-3,0)
    .. controls (-2.55,1.2) and (-1.5,2.2) .. (v)
    .. controls (1.5,2.2) and (2.55,1.2) .. (3,0)
    -- (t);
  \begin{scope}[blue!75!black,line width=1.7pt,
    dash pattern=on 5pt off 3pt,dash phase=1pt]
    \draw (s) -- (-3,0);
    \draw (t) -- (3,0);
    \foreach \x in {-3,3}
      \draw (v)
        .. controls ({0.5*\x},2.2) and ({0.85*\x},1.2) .. (\x,0);
  \end{scope}

  \foreach \p in {s,t,v}
    \fill (\p) circle (1.6pt);
  \foreach \x in {-3,-1,1,3}
    \fill (\x,0) circle (1.6pt);
  \node[above] at (v) {$x$};
  \node[left] at (s) {$s$};
  \node[right] at (t) {$t$};
  \path (-3,0)
    .. controls (-2.55,1.2) and (-1.5,2.2) .. (v)
    node[blue!75!black,midway,left=5pt] {$P'$};
  \node[red,right] at (3.4,2) {$B(x, R)$};
\end{tikzpicture}

\caption{Changing the $s$-$t$ path being considered to ``shortcut'' through the cluster $C$ upper bounds the distance increase due to the cluster by $O(D)$, since the cluster has diameter at most $D$.}

\label{fig:octopus-2}
\end{figure}

Fix a shortest $s$-$t$ path $P$, and consider a clustering $\mathcal{C}$. For each cluster cutting the path, suppose we either pay for each edge of the path cut by the cluster, or only for the first and last such edge plus the cost to shortcut through the cluster. Then, the distance increase between $s$ and $t$ is at most this total cost. Indeed, a simple greedy algorithm can construct an $s$-$t$ path in $\widetilde{G}$ with at most this increase in weight: walk edges of the path $P$ one at a time, except that whenever you enter a cluster $C$ for which the shortcut cost was paid, take a path of length at most $\diam(G[C]) \leq D$ in the interior of that cluster to the last vertex on $P$ inside that cluster. \Cref{clm:shortcutting} formalizes a slightly more general statement. The proof is deferred to \Cref{sec:deferred}.
\begin{restatable}[Shortcutting]{claim}{shortcutting}\label{clm:shortcutting}
Let $\Delta$ be edge weight increases, $\widetilde{G} \coloneqq (V, E, \widetilde{w} \coloneqq w + \Delta)$, and $\mathcal{C} = \{C_1, \dots, C_k\}$ be a clustering. 

For a pair of vertices $u, v \in V$ and a path $P$ from $u$ to $v$, let $\{P_1,P_2,\dots,P_k\}$ be a partition of the edges of $P$ so that edges contained in a cluster $C_i$ are assigned to that cluster's part $P_i$, and edges between two clusters are assigned arbitrarily to either cluster's part. Then,
\begin{equation*}
    \dist_{\widetilde G}(u,v) \le w(P) + \sum_i\min\{\Delta(P_i), D + 2M\}
\end{equation*}
where $D \coloneqq \max_{C \in \mathcal{C}} \diam(\widetilde{G}[C])$ is the maximum cluster diameter in $\widetilde{G}$ and $M \coloneqq \max_{e \in \delta(\mathcal{C})} \Delta(e)$ the maximum weight increase of any cluster boundary edge.
\end{restatable} 

This motivates an analysis considering the costs of clusters one by one, for which the simplest LDD algorithm to consider is geometric ball carving.

\vspace{1em}
\begin{algorithm}[H]
\caption{\textsc{GEOM}($G$, $D$)}
\label{alg:geom}
$p \gets \frac{c_{geom}\log n}{D}$\\
$U \gets V$\\
$\mathcal{C} \gets \emptyset$\\
\While{$U \neq \emptyset$}{
    $x \gets $ an arbitrary element of $U$\\
    Sample $R \sim \mathrm{Geom}\left(p\right)$\\
    $C \gets B_{G[U]}(x,R)$\\
    $U \gets U \setminus C$\\
    $\mathcal{C} \gets \mathcal{C} \cup \{C\}$\\
}
\Return $\mathcal{C}$\\
\end{algorithm}
\vspace{1em}

As usual for an analysis of $\textsc{GEOM}$, we consider an alternate process 
producing the same distribution of clusterings. This process proceeds in discrete time steps $t$, and maintains an active ball $B_t$ in the remaining graph $G[U]$ by center $x_t$ and radius $R_t$. Each time step, it either carves out the active ball with probability $p$ and starts a new ball in the remaining graph, or otherwise increases the radius of the active ball by $1$.

We say an edge $e$ is \textit{threatened} by the active ball at time $t$ if $e \in \delta_{G[U]}(B_t)$. The standard analysis of GEOM proves that any fixed edge $e$ is cut with probability at most $p \cdot w(e)$ by noting that any edge is threatened for at most $w(e)$ time steps, which must be consecutive, as once the active ball threatens an edge, it will either cut it or grow to contain it in $w(e) - 1$ time steps.

If a ball was cut at time $t$, \Cref{clm:shortcutting} gives us that the cost of the produced cluster is at most $\min(\epsilon D \cdot X_t, 3D)$, where $X_t$ denotes the number of edges being threatened at time $t$. By the standard analysis, $\sum_t X_t \leq w(P)$, thus the expected total cost is at most $p \cdot \epsilon D \cdot w(P) = O(\epsilon \log n) \cdot w(P)$. But now, since the cost of any individual time step is bounded, we can use a martingale concentration inequality to obtain a high probability bound on the total cost significantly exceeding the expectation.

With this approach, we obtain in \Cref{sec:geom-analysis} the following:

\begin{lemma}[Informal version of \Cref{lem:geom-distance-preserving}]\label{lem:geom-distance-preserving-informal}
    For $\mathcal{C} = \textsc{GEOM}(G, D)$, with high probability,
    \begin{equation*}
        \dist_{\widetilde{G}}(u, v) \leq (1 + O(\epsilon \log n)) \dist_G(u, v) + O(D \log n).
    \end{equation*}
\end{lemma}

Comparing this to \Cref{eq:concentration-distance-independent}, we note that the only difference is that the additive term is not scaled down by $\epsilon$.

We also obtain in \Cref{sec:mpx-analysis} a slightly weaker result for the $\textsc{MPX}$ LDD algorithm \cite{miller2013mpxoriginal}. We analyze an alternative process producing the same distribution of clusterings that grows the clusters one-by-one in random order. We use the fact that in this process, each vertex changes the cluster it is in at most $O(\log n)$ times with high probability, thus each edge is threatened at most $O(w(e) \log n)$ times. The rest of the proof proceeds similarly to the proof for $\textsc{GEOM}$, and gives:

\begin{lemma}[Informal version of \Cref{lem:mpx-distance-preserving}]\label{lem:mpx-distance-preserving-informal}
    For $\mathcal{C} = \textsc{MPX}(G, D)$, with high probability,
    \begin{equation*}
        \dist_{\widetilde{G}}(u, v) \leq (1 + O(\epsilon \log^2 n)) \dist_G(u, v) + O(D \log n).
    \end{equation*}
\end{lemma}

\paragraph{Diameter Preserving Hierarchies.} We finally return to considering a hierarchy $\mathcal{C}_L, \dots, \mathcal{C}_0$ of low-diameter decompositions computed top-down using an LDD algorithm that satisfies the bound of \Cref{lem:mpx-distance-preserving-informal} and has slack $O(\log n)$, such as \textsc{GEOM} or \textsc{MPX}.\footnote{The stronger bound of \Cref{lem:geom-distance-preserving-informal} for GEOM would not lead to a substantially weaker requirement on $\epsilon$, as we need edge lengths in $\widetilde{G}$ to increase at most by a multiplicative constant in expectation compared to $G$, which already requires $\epsilon = O(1 / L \log(n))$.} We sketch a proof that such a hierarchy will with high probability satisfy $\diam(\widetilde{G}[C]) \leq 7 D_\ell$ for every $\ell \in [L]$ and $C \in \mathcal{C}_\ell$, for small enough $\epsilon = \Theta(1 / \log^2 n)$ and $\alpha = \Theta(\log n)$. The corresponding full proof is given in \Cref{sec:hierarchy-diameter-analysis}.

Fix $\mathcal{C}_L, \dots, \mathcal{C}_\ell$, and consider some cluster $C \in \mathcal{C}_\ell$, which has $\diam(G[C]) \leq D_\ell$. When $\mathcal{C}_{\ell - 1}$ is computed, denote by $\widetilde{G}_{\ell - 1}$ the graph with only weight increases from level $\ell - 1$. Since the LDD algorithm is distance preserving, we have for small enough $\epsilon = \Theta(1 / \log^2 n)$ and large enough $\alpha = \Theta(\log n)$ with high probability that
\begin{align*}
    \diam(\widetilde{G}_{\ell - 1}[C])
        &\leq (1 + O(\epsilon \log^2 n)) D_\ell + O(D_{\ell - 1} \log n)\\
        &\leq 3 D_\ell,
\end{align*}
thus the effect on diameter from the level immediately below $\ell$ is at most a constant factor.

Consider now two fixed vertices $u, v \in C$ and a shortest path $P$ in $\widetilde{G}_{\ell - 1}$ between them. Then,
\begin{itemize}
    \item For small enough $\epsilon = \Theta(1 / \log^2 n)$, the expected weight increase from the lower levels to $P$ is at most $w(P) \leq \widetilde{w}_{\ell - 1}(P) \leq 3D_\ell$, as there are at most $L = O(\log n)$ such levels, and the expected contribution of each is at most $O(\epsilon \log n) \cdot w(P)$, as the LDD has slack $O(\log n)$.
    \item The weight increases inside distinct clusters of $\mathcal{C}_{\ell - 1}$ are independent.
\end{itemize}
To obtain a concentration bound, we need an upper bound on the maximum increase caused by weight increases from any single cluster $C' \in \mathcal{C}_{\ell - 1}$. For the \emph{fixed} path $P$, this is impossible. However, we can once again apply shortcutting (\Cref{clm:shortcutting}) to bound the increase to the $s$-$t$ \emph{distance}. If $\diam(\widetilde{G}[C']) \leq 7 D_{\ell - 1}$ holds for every $C' \in \mathcal{C}_{\ell - 1}$ (which can be taken as the assumption of an induction on increasing $\ell$), then the maximum contribution of any single cluster is $7 D_{\ell - 1} = O(\frac{1}{\alpha} D_\ell)$. Since the contributions are independent, their expected sum is small, and any individual contribution is an arbitrarily small $\Theta(\frac{1}{\log n})$ fraction of the expectation for large enough $\alpha = \Theta(\log n)$, a standard Chernoff bound proves that with high probability, $\dist_{\widetilde{G}[C]}(u, v) \leq 7 D_\ell$. A union bound over pairs $u, v$ completes the induction step.

\section{Preliminaries} \label{Preliminaries}

\paragraph{Basic Notation.} We work with a connected, simple input graph $G=(V,E, w)$ with positive, polynomially bounded, integral edge weights $w : E \to [W]$. For a path $P$, we write $w(P) \coloneqq \sum_{e \in P} w(e)$.

We denote the ball of radius $R$ centered at $v$ in $G$ by
$B_{G}(v,R) = \{u \in V \mid \dist_{G}(u,v) \le R\}$.
The diameter of $G$ is defined as
\begin{equation*}
\diam(G) \coloneqq \max_{u, v \in V} \dist_{G}(u, v).    
\end{equation*}
Moreover, given a set $S \subseteq V$, we denote by $G[S]$ the induced subgraph of $G$ on $S$. We denote 
\begin{equation*}
\diam_{G}(S) \coloneqq \max_{u, v \in S} \dist_{G}(u, v).    
\end{equation*}
We call $\diam_{G}(S)$ the weak diameter of $S$ in $G$ and $\diam(G[S])$ the strong diameter of $S$. We denote the set of edges with one endpoint inside $S$ and the other inside $V\setminus S$ by $\delta(S) = \delta_G(S)$. For a clustering $\mathcal{C}$, we write $\delta(\mathcal{C}) \coloneqq \bigcup_{C \in \mathcal{C}} \delta(C)$.

\paragraph{The Geometric and Exponential Distributions.}
For a parameter $p \in (0,1)$, the \textit{geometric distribution} $X \sim \mathrm{Geom}(p)$ takes values in $\{0,1,2,\ldots\}$ with probabilities
\begin{equation*}
\mathbb{P}(X = k) = (1-p)^k p, \qquad k = 0,1,2,\ldots.
\end{equation*}
Equivalently, $X$ counts the number of failures before the first success in a sequence of independent Bernoulli trials with success probability $p$. For simplicity of notation, we also extend the definition of $\mathrm{Geom}$ to $p \geq 1$ with such $\mathrm{Geom}(p)$ being identically zero.

For a parameter $\lambda > 0$, the \textit{exponential distribution} $X \sim \mathrm{Exp}(\lambda)$ takes values in $\mathbb{R}_{\geq 0}$ with probability density function
\begin{equation*}
p_X(x) = \lambda e^{-\lambda x}, \quad x \geq 0
\end{equation*}
and cumulative distribution function
\begin{equation*}
F_X(x) = 1 - e^{-\lambda x}, \quad x \geq 0.
\end{equation*}

Both the geometric distribution and the exponential distribution have the \emph{memoryless property}: 
for $X \sim \mathrm{Geom}(p)$, for all integers $k, \ell \geq 0$,
\begin{equation*}
\mathbb{P}(X \geq k + \ell \mid X \geq \ell) = \mathbb{P}(X \geq k),
\end{equation*}
and for $X \sim \mathrm{Exp}(\lambda)$, for all $x, y > 0$,
\begin{equation*}
\mathbb{P}(X > x + y \mid X > y) = \mathbb{P}(X > x).    
\end{equation*}

The integer component of an exponential random variable is geometrically distributed. Specifically, the following holds:
\begin{observation}\label{obs:floor-of-exp}
    If $X \sim \mathrm{Exp}(\lambda)$, then $\lfloor X \rfloor \sim \mathrm{Geom}(p)$ for $p = 1 - e^{-\lambda}$.
\end{observation}
\begin{proof}
    $\lfloor X \rfloor = k$ occurs with probability
    \begin{equation*}
        \mathbb{P}(\lfloor X \rfloor = k) = F_X(k + 1) - F_X(k) = e^{-\lambda k} - e^{-\lambda (k + 1)} = e^{-\lambda k} (1 - e^{-\lambda}) = (1 - p)^k p.
    \end{equation*}
    This is the definition of a geometric random variable with parameter $p \coloneqq 1 - e^{-\lambda}$.
\end{proof}

\paragraph{Concentration Bounds.} We use the following  classic Chernoff bound. 

\begin{theorem}[Chernoff bound]\label{thm:chernoffBound}
Let $Y_1,\ldots,Y_n\in[0,H]$ be independent random variables, let
\begin{equation*}
Y = \sum_i Y_i,    \qquad \mu=\mathbb{E}[Y].
\end{equation*}
Then, for every $\delta \geq 0$,
\begin{equation*}
\mathbb{P}(Y \geq (1 + \delta)\mu) \leq \exp\left(-\frac{\delta^2 \mu}{(2 + \delta) H}\right).
\end{equation*}
\end{theorem}
For martingales, we use the following bound, originally due to Freedman \cite{freedman1975tail}.
\begin{theorem}[Freedman's inequality {\cite[Thm. 1.1]{tropp11exactfreedman}}] \label{FreedmanInequality}
Let $\{M_t\}_{t \ge 0}$ be a real-valued martingale with $M_0 = 0$ and difference sequence $Y_t \coloneqq M_t - M_{t-1}$,
and suppose that $Y_t \le R$ almost surely for every $t \ge 1$. Define the predictable quadratic variation process of the martingale:
\begin{equation*}
W_t \coloneqq \sum_{j=1}^t \mathbb E_{j-1}[Y_j^2].    
\end{equation*}
Then, for every $k \ge 0$ and $\sigma^2 > 0$,
\begin{equation*}    
\mathbb{P}\left[\exists t \ge 0 \mid M_t \ge k\ \text{and}\ W_t \le \sigma^2\right]
\leq \exp\left(-\frac{k^2/2}{\sigma^2 + Rk/3}\right).
\end{equation*}
\end{theorem}
\section{Distance Preserving LDDs}

In this section, we prove concentration bounds on distances in a graph perturbed by a single clustering produced by either \textsc{GEOM} or \textsc{MPX}. We desire guarantees of the following kind:

\begin{definition}[distance preserving LDD algorithm]\label{def:distance-preserving}
    An LDD algorithm is \textit{$(\beta, \gamma)$-distance preserving}, if it has slack at most $\beta$, and for any diameter bound $D$ and fixed $\epsilon \in [0, 1]$, the random graph $\widetilde{G} = (V, E, \widetilde{w})$ constructed by computing $\mathcal{C} = \textsc{LDD}(G, D)$ and increasing the weight of edges cut by $\mathcal{C}$ by $\epsilon D$ satisfies with high probability that for every pair $u, v \in V$,
    \begin{equation*}
        \dist_{\widetilde{G}}(u, v) \leq (1 + \epsilon \beta) \dist_{G}(u, v) + \gamma D.
    \end{equation*}
\end{definition}

We will prove in \Cref{sec:geom-analysis} for $\textsc{GEOM}$ and \Cref{sec:mpx-analysis} for $\textsc{MPX}$ the following:

\begin{restatable}{lemma}{geomdistancepreserving}\label{lem:geom-distance-preserving}
    GEOM is $(\beta, \gamma)$-distance preserving for $\beta, \gamma = O(\log n)$.
\end{restatable}

\begin{restatable}{lemma}{mpxdistancepreserving}\label{lem:mpx-distance-preserving}
    MPX is $(\beta, \gamma)$-distance preserving for $\beta = O(\log^2 n)$ and $\gamma = O(\log n)$.
\end{restatable}

\subsection{GEOM is Distance Preserving}\label{sec:geom-analysis}

The geometric ball carving algorithm (\Cref{alg:geom}) and the approach we will use to show it is distance preserving have already been described in the technical overview. The equivalent variant of \textsc{GEOM} we analyze is given in \Cref{alg:geom2}.

\begin{algorithm}[h]
\caption{\textsc{GEOM2}($G$, $D$)}
\label{alg:geom2}
$p \gets \min\left(1, \frac{c_{\text{geom}} \log n}{D}\right)$\\
$U \gets V$\\
$\mathcal{C} \gets \emptyset$\\
$x \gets$ an arbitrary element of $U$\\
$R \gets 0$\\
\For{$t = 1, 2, 3, \dots$}{
    Sample $\cut \sim \mathrm{Bernoulli}(p)$\\
    \If{not $\cut$}{
        $R \gets R + 1$\\
    }\Else{
        $C \gets B_{G[U]}(x,R)$\\
        $U \gets U \setminus C$\\
        $\mathcal{C} \gets \mathcal{C} \cup \{C\}$\\
        \If{$U = \emptyset$}{
            \Return $\mathcal{C}$
        }\Else{
            $x \gets$ an arbitrary element of $U$\\
            $R \gets 0$\\
        }
    }
}
\end{algorithm}

We will need the following standard lemma, which we prove in \Cref{sec:deferred} for completeness.
\begin{restatable}[\cite{bartal1996probabilistic}]{lemma}{geom}
\label{lem:geom}
GEOM is an LDD algorithm with slack $s \coloneqq c_{geom} \log n$.
\end{restatable}

When analyzing \Cref{alg:geom2}, we refer to iterations of the main for-loop as time steps, denote by $\cut_t$ the value of the cut parameter sampled at time $t$, by $x_t$ and $R_t$ the values of $x$ and $R$ at the start of the $t$\textsuperscript{th} time step, and by $B_t \coloneqq B_{G[U]}(x_t, R_t)$ the \textit{active ball} at time $t$. We say the active ball \textit{threatens} an edge $e$ if $e \in \delta_{G[U]}(B_t)$.



\geomdistancepreserving*

\begin{proof}
    The case $p = 1$ is trivial. Assume $p < 1$, fix two vertices $u, v$, and let $P$ be a shortest $(u, v)$-path in $G$.

    Let $X_t$ be a random variable denoting the number of edges of $P$ the active ball at time $t$ threatens. Since each edge $e$ is threatened during at most $w(e)$ time steps, we have $\sum X_t \leq w(P)$.
    
    Let $Y_t \coloneqq \mathbb{I}[\text{cut}_t] \cdot \min(\epsilon D \cdot X_t, 3D)$, where $\text{cut}_t$ is the event that the active ball stops growing at time $t$, thus becoming one of the produced clusters.

    Conditioning on every produced cluster having diameter at most $D$, an application of \Cref{clm:shortcutting} shows that $\dist_{\widetilde G}(u,v) \le w(P)+\sum_t Y_t$. Indeed, we can assign each perturbed edge of $P$ to the first cluster that cut it, and each unperturbed edge to the cluster it is contained in. For a cluster produced at time $t$, the assigned edges have a total weight increase of $\epsilon D \cdot X_t$, and each edge's weight increases by at most $\epsilon D \le D$, thus shortcutting through a cluster costs $D + 2D = 3D$.

    We can bound the expectation of the sum $\sum_t Y_t$ by bounding each minimum term by its first argument. To obtain a concentration bound, we will use Freedman's inequality (\Cref{FreedmanInequality}) on the martingale $M_t = M_{t - 1} + Z_t$, where $M_0 = 0$ and $Z_t$ is the centered version of $Y_t$, specifically
    \begin{equation*}
        Z_t \coloneqq Y_t - \mathbb{E}[Y_t \mid X_t] = \min\left(\epsilon D \cdot X_t, 3D\right) \cdot \begin{cases}
            (1 - p) &\text{if $\cut_t$}\\
            -p &\text{otherwise}
        \end{cases}.
    \end{equation*}
    For variance, a direct calculation gives
    \begin{align*}
        \mathbb{E}[Z_t^2 \mid X_t]   &= \min\left(\epsilon D \cdot X_t, 3D\right)^2 \cdot (p (1 - p)^2 + p^2 (1 - p))\\
            &\leq \epsilon D \cdot X_t \cdot 3D \cdot p (1 - p)\\
            &\leq 3 \epsilon Ds \cdot X_t.
    \end{align*}
    where $s$ is the slack of \textsc{GEOM}. Since the distribution of $Z_t$ is determined by $X_t$, we thus have
    \begin{align*}
        W_t &= \sum\nolimits_{j = 1}^{t} \mathbb{E}[Z_j^2 \mid X_j]\\
            &\leq 3 \epsilon Ds \sum\nolimits_{j = 1}^{t} X_j\\
            &\leq 3 \epsilon Ds \cdot w(P).
    \end{align*}
    Let $R \coloneqq 3 D$ and $\sigma^2 \coloneqq 3 \epsilon Ds \cdot w(P)$. As $Z_t \leq R$ and $W_t \leq \sigma^2$ holds for every $t$, by Freedman,
    \begin{align*}
        \mathbb{P}\left(\sum\nolimits_t Z_t \geq k\right)    &\leq \mathbb{P}(\exists t : M_t \geq k)\\
            &= \mathbb{P}(\exists t : M_t \geq k \text{ and } W_t \leq \sigma^2)\\
            &\leq \exp\left(\frac{-k^2 / 2}{\sigma^2 + Rk / 3}\right).
    \end{align*}
    Plugging in $k = 3 \sigma^2 / R + \gamma D$ for $\gamma = \Theta(\log n)$ to be determined later, the larger term in the divisor is $Rk / 3 = Dk$, and we obtain
    \begin{equation*}
        \mathbb{P}\left(\sum\nolimits_t Z_t \geq k\right)
            \leq \exp\left(\frac{-k^2 / 2}{2 Dk}\right)
            = \exp\left(\frac{-k}{4D}\right)
            \leq \exp(-\gamma / 4).
    \end{equation*}
    For large enough $\gamma = \Theta(\log n)$, we get that $\sum_t Z_t \leq k$ with high probability. Finally, we have
    \begin{equation*}
        \sum_t \mathbb{E}[Y_t \mid X_t] \leq \sum_t \mathbb{E}[\mathbb{I}[\cut_t] \cdot \epsilon D \cdot X_t \mid X_t] = \epsilon p D \cdot \sum_t X_t \leq \epsilon s \cdot w(P).
    \end{equation*}
    Thus, conditioning on both $\sum_t Z_t \leq k$ and each cluster having diameter at most $D$, which are both individually high probability events, we obtain that with high probability,
    \begin{align*}
        \dist_{\widetilde{G}}(u, v)
            &\leq w(P) + \sum\nolimits_t Y_t\\
            &= w(P) + \sum\nolimits_t Z_t + \sum\nolimits_t \mathbb{E}[Y_t \mid X_t]\\
            &\leq w(P) + k + \epsilon s \cdot w(P)\\
            &= (1 + \epsilon \beta) \cdot w(P) + \gamma D\\
            &= (1 + \epsilon \beta) \dist_G(u, v) + \gamma D
    \end{align*}
    for some $\beta = O(s) = O(\log n)$. Union bounding over all $\Theta(n^2)$ pairs $u, v$ gives the desired result.
\end{proof}

\subsection{MPX is Distance Preserving}\label{sec:mpx-analysis}

\begin{algorithm}[h]
\caption{\textsc{MPX}($G$, $D$)}
\label{alg:MPX}
\For{$v \in V$}{
    Sample $R_v \sim \mathrm{Exp}\left(\frac{c_{\text{mpx}} \log n}{D}\right)$
}
\For{$u \in V$}{
    $r_u \gets \arg\min_{v \in V} \dist_G(u, v) - R_v$
}
$\mathcal{C} \gets \emptyset$\\
\For{$v \in V$}{
    \If{$r_v = v$}{
        $\mathcal{C} \gets \mathcal{C} \cup \{\{u \in V \mid r_u = v\}\}$
    }
}
\Return $\mathcal{C}$
\end{algorithm}

The MPX low diameter decomposition \cite{miller2013mpxoriginal}, given in \Cref{alg:MPX}, samples for each vertex an exponentially distributed ``head start'' $R_v$, then assigns each vertex $u$ to the cluster with center $v$ of minimum $\dist(u, v) - R_v$. 

\begin{restatable}[\cite{miller2013mpxoriginal}]{lemma}{mpxldd}\label{lem:mpx-ldd}
    MPX is an LDD algorithm with slack $s \coloneqq 2c_{\text{mpx}} \log n$.
\end{restatable}

To show that MPX is distance-preserving, we analyze a variant that grows the clusters one-by-one and has the same distribution of produced clusterings $\mathcal{C}$, given in \Cref{alg:MPX2}. 

\begin{algorithm}[h]
\caption{\textsc{MPX2}($G$, $D$)}
\label{alg:MPX2}
$p \gets 1 - \exp(\frac{- c_{\text{mpx}} \log n}{D})$\\
$R_v \gets 0$ for every $v \in V$\\
$C_v \gets \emptyset$ for every $v \in V$\\
$v_1, v_2, \dots, v_n \gets$ a random permutation of $V$\\
$i \gets 1$\\
\For{$t = 1, 2, 3, \dots$}{
    Sample $\cut \sim \mathrm{Bernoulli}(p)$\\
    \If{not $\cut$}{
        $R_{v_i} \gets R_{v_i} + 1$\\
    }\Else{
        $C_{v_i} \gets \{u \in V \mid \dist(u, v_i) - R_{v_i} < \min_{j < i} \dist(u, v_j) - R_{v_j}\}$\\
        \For{$j < i$}{
            $C_{v_j} \gets C_{v_j} \setminus C_{v_i}$\\
        }
        \If{$i = n$}{
            $\mathcal{C} \gets \{C_v \mid v \in V, C_v \neq \emptyset\}$\\
            \Return $\mathcal{C}$\\
        }\Else{
            $i \gets i + 1$\\
        }
    }
}
\end{algorithm}

\begin{lemma}
    The randomness of \Cref{alg:MPX} can be coupled with the randomness of \Cref{alg:MPX2} so that the produced clusterings are almost surely equal. 
\end{lemma}
\begin{proof}
    By \Cref{obs:floor-of-exp}, since $R \sim \mathrm{Exp}(\lambda)$ for $\lambda \coloneqq \frac{c_{\text{mpx}} \log n}{D}$, $\lfloor R \rfloor \sim \mathrm{Geom}(p)$ for $p \coloneqq 1 - e^{-\lambda} = 1 - \exp(\frac{-c_{\text{mpx}} \log n}{D})$. Thus, we can couple the randomness that determines the random variables $R_v$ so that $\lfloor R_v \rfloor$ in \Cref{alg:MPX} equals $R_v$ in \Cref{alg:MPX2} almost surely.
    
    To couple the tiebreaking, we let the random permutation in \Cref{alg:MPX2} be sorted in decreasing order by the fractional part of the variables $R_v$ in \Cref{alg:MPX}.
\end{proof}

We refer to iterations of the main for-loop in \Cref{alg:MPX2} as time steps, denote by $\cut_t$ the value of the cut parameter sampled at time $t$ and by $B_t \coloneqq \{u \in V \mid \dist(u, v_i) - R_{v_i} < \min_{j < i} \dist(u, v_j) - R_{v_j}\}$ the \textit{active ball} at time $t$ that would be assigned to $C_{v_i}$ if a cut was made that time step. We say the active ball \textit{threatens} an edge $e$ if $e \in \delta(B_t)$.

The analysis mostly proceeds similarly to the analysis of GEOM from \Cref{sec:geom-analysis}. In the analysis of GEOM, we use the fact that each edge $e \in P$ is threatened during at most $w(e)$ time steps, thus $\sum_t X_t \leq w(P)$ always holds, where $X_t$ is the number of edges threatened in time step $t$. This is not true for MPX, as future clusters can grow over existing ones.

To obtain a bound on $\sum_t X_t$, we exploit the fact that when the clusters are grown in random order, each vertex changes clusters at most $O(\log n)$ times with high probability. To see this, fix the sampled head starts $R_v$, then randomly permute the vertices. The number of times a fixed vertex $u$ changes clusters now equals the number of (strict) prefix minimums of $\dist(u, v_i) - R_{v_i}$, which is $O(\log n)$ with high probability. We prove \Cref{lem:rand-perm-pref-mins} in \Cref{sec:deferred} for completeness.

\begin{restatable}{lemma}{randpermprefmins}\label{lem:rand-perm-pref-mins}
    Let $\pi$ be a random permutation of $[n] = \{1, 2, \dots, n\}$. Call $\pi_i$ a \textit{prefix minimum} if $\pi_i < \min_{j < i} \pi_j$. Then, the number of prefix minima is $O(\log n)$ with high probability.
\end{restatable}

Now, assume each vertex changes clusters at most $O(\log n)$ times. For an edge to be threatened while the $i$\textsuperscript{th} cluster is being grown, at least one of its two endpoints must have joined that cluster. As that active ball still threatens it for at most $2w(e)$ time steps, it must be threatened for at most $O(\log(n)) \cdot w(e)$ time steps in total. Thus, $\sum_t X_t = O(\log n) \cdot w(P)$.



\mpxdistancepreserving*

\begin{proof}
    Fix two vertices $u, v$, and let $P$ be a shortest $(u, v)$-path in $G$.

    For a vertex $u$, say $u$ is \textit{touched} by the $i$\textsuperscript{th} cluster if $u$ is in $C_i$ when it is initially created. Let $I_{\text{touch}}$ be the event that no vertex is touched more than $c_{\text{pref}} \log n$ times. Then, by \Cref{lem:rand-perm-pref-mins}, $I_{\text{touch}}$ occurs with high probability for sufficiently large constant $c_{\text{pref}}$.

    Let $X_t$ be a random variable denoting the number of edges of $P$ the active ball at time $t$ threatens. Given $I_{\text{touch}}$, we have $\sum_t X_t \leq 4c_{\text{pref}} \log(n) \cdot w(P)$, as an edge $e$ can only be threatened by the active ball while the $i$\textsuperscript{th} cluster is being grown if at least one of its endpoints is touched by the $i$\textsuperscript{th} cluster, and $e$ is threatened for at most $2w(e)$ time steps of that cluster's growth. 

    Let $Y_t \coloneqq \mathbb{I}[\text{cut}_t] \cdot \min(\epsilon D \cdot X_t, 3D)$. Conditioning on every cluster in the final clustering $\mathcal{C}$ having diameter at most $D$, an application of \Cref{clm:shortcutting} shows that $\dist_{\widetilde{G}}(u,v) \leq w(P) + \sum_t Y_t$. Indeed, we can assign each perturbed edge of $P$ to the \emph{last} cluster that cut it, and each unperturbed edge to the cluster it is contained in. For a cluster produced at time $t$, the assigned edges have a total weight increase of \emph{at most} $\epsilon D \cdot X_t$, and each edge's weight increases by at most $\epsilon D \leq D$, thus shortcutting through a cluster costs $D + 2D = 3D$.

    We can bound the expectation of the sum $\sum_t Y_t$ by bounding each minimum term by its first argument.
    To obtain a concentration bound, we will use Freedman's inequality (\Cref{FreedmanInequality}) on the martingale $M_t = M_{t - 1} + Z_t$, where $M_0 = 0$ and $Z_t$ is the centered version of $Y_t$, specifically
    \begin{equation*}
        Z_t \coloneqq Y_t - \mathbb{E}[Y_t \mid X_t] = \min\left(\epsilon D \cdot X_t, 3D\right) \cdot \begin{cases}
            (1 - p) &\text{if $\cut_t$}\\
            -p &\text{otherwise}
        \end{cases}.
    \end{equation*}
    For variance, a direct calculation gives
    \begin{align*}
        \mathbb{E}[Z_t^2 \mid X_t]   &= \min\left(\epsilon D \cdot X_t, 3D\right)^2 \cdot (p (1 - p)^2 + p^2 (1 - p))\\
            &\leq \epsilon D \cdot X_t \cdot 3D \cdot p (1 - p)\\
            &\leq 3 c_{\text{mpx}} \cdot \epsilon D \log(n) \cdot X_t
    \end{align*}
    where we use that $p = 1 - \exp(\frac{-c_{mpx} \log n}{D}) \leq \frac{ c_{mpx} \log n}{D}$ since $e^{-x} \geq 1 - x$. As the distribution of $Z_t$ is completely determined by $X_t$, we thus have
    \begin{equation*}
        W_t = \sum\nolimits_{j = 1}^{t} \mathbb{E}[Z_j^2 \mid X_j] \leq 3 c_{\text{mpx}} \cdot \epsilon D \log(n) \cdot \sum_{j = 1}^{t} X_j
    \end{equation*}
    and, given $I_{\text{touch}}$ occurs, for every $t$,
    \begin{equation*}
        W_t \leq 12 c_{\text{mpx}} c_{\text{pref}} \cdot \epsilon D \log^2(n) \cdot w(P).
    \end{equation*}
    Let $R \coloneqq 3 D$ and $\sigma^2 = O(\epsilon D \log^2 n) w(P)$ be the right-hand side of the above inequality. As $Z_t \leq R$ holds for every $t$, by Freedman,
    \begin{align*}
        \mathbb{P}\left(\sum\nolimits_t Z_t \geq k\right)    &\leq \mathbb{P}(\exists t : M_t \geq k)\\
            &\leq \mathbb{P}(\exists t : M_t \geq k \text{ and } W_t \leq \sigma^2) + \mathbb{P}(\overline{I_{\text{touch}}})\\
            &\leq \exp\left(\frac{-k^2 / 2}{\sigma^2 + Rk / 3}\right) + \mathbb{P}(\overline{I_{\text{touch}}}).
    \end{align*}
    Plugging in $k = 3 \sigma^2 / R + \gamma D$ for $\gamma = \Theta(\log n)$ to be determined later, the larger term in the divisor is $Rk / 3 = Dk$, and we can bound the exponential by
    \begin{equation*}
        \exp\left(\frac{-k^2 / 2}{\sigma^2 + Rk / 3}\right) \leq \exp\left(\frac{-k^2 / 2}{2Dk}\right) = \exp\left(\frac{-k}{4D}\right) \leq \exp(-\gamma / 4)
    \end{equation*}
    thus for large enough $\gamma = \Theta(\log n)$, we get that $\sum_t Z_t \leq k$ with high probability. Finally, we have
    \begin{equation*}
        \sum_t \mathbb{E}\left[Y_t \mid X_t\right]
            \leq \sum_t \mathbb{E}\left[\mathbb{I}[\cut_t] \cdot \epsilon D \cdot X_t \mid X_t\right]
            \leq \epsilon p D \cdot \sum_t X_t.
    \end{equation*}
    Thus, conditioning on $I_{\text{touch}}$, $\sum_t Z_t \leq k$, and each cluster having diameter at most $D$, which are each individually high probability events, we obtain that with high probability,
    \begin{align*}
        \dist_{\widetilde{G}}(u, v)
            &\leq w(P) + \sum\nolimits_t Y_t\\
            &= w(P) + \sum\nolimits_t Z_t + \sum\nolimits_t \mathbb{E}[Y_t \mid X_t]\\
            &\leq w(P) + k + \epsilon p D \cdot 4 c_{\text{pref}} \log(n) \cdot w(P)\\
            &= (1 + \epsilon \beta) \cdot w(P) + \gamma D\\
            &= (1 + \epsilon \beta) \dist_G(u, v) + \gamma D.
    \end{align*}
    for some $\beta = O(\log^2 n)$. Union bounding over all $\Theta(n^2)$ pairs $u, v$ gives the desired result.
\end{proof}

\section{Diameter Preserving Hierarchies}\label{sec:hierarchy-diameter-analysis}

\begin{algorithm}[H]
\caption{\textsc{LDDHierarchy}($G$, $\alpha$, $\textsc{LDD}$)}
$L \gets \lceil \log_{\alpha} (nW) \rceil$\\
$D_\ell \gets \alpha^\ell$ for every $\ell \in \{0, 1, \dots, L\}$\\
$\mathcal C_L \gets \{V\}$\\
\For{$\ell = L-1, \ldots, 1$}{
    $G_\ell \gets G \setminus \delta(\mathcal{C}_{\ell + 1})$\\
    $\mathcal{C}_\ell \gets \textsc{LDD}(G_\ell, D_\ell)$
}
$\mathcal C_0 \gets \{\{v\}\}_{v \in V}$\\
\Return{$\{\mathcal{C}_\ell\}_\ell$}
\end{algorithm}
\vspace{0.5cm}

In this section, we prove diameter guarantees \emph{in the perturbed graph} for LDD hierarchies built with distance preserving LDD algorithms. We desire the following kind of guarantee:
\begin{definition}[Diameter Preserving LDD Hierarchy]\label{def:diameter-preserving}
    For an LDD hierarchy $\{\mathcal{C}_\ell\}_{\ell}$ on a graph $G = (V, E, w)$ and a parameter $\epsilon \geq 0$, let $\widetilde{G} = (V, E, \widetilde{w})$ be $G$ with edge weights $\widetilde{w}(e) \coloneqq w(e) + \epsilon D_{\ell_e}$, where $\ell_e \coloneqq \max \{\ell : e \in \delta(\mathcal{C}_\ell)\}$ is the highest level that cuts $e$, and $D_{\ell_e}$ is its diameter bound.
    
    An LDD hierarchy is \textit{$(\epsilon,c)$-diameter preserving}, if $\diam(\widetilde{G}[C]) \leq c D_\ell$ for each $\ell \in [L]$ and $C \in \mathcal{C}_\ell$, i.e. if the diameter guarantee for every cluster worsens by at most a multiplicative $c$-factor.
\end{definition}

We obtain the following. Combined with the previous section, it implies as a corollary that hierarchies built with \textsc{GEOM} or \textsc{MPX} preserve diameters up to a constant factor for any noise magnitude $\epsilon = O(1 / \log^2 n)$, as claimed in the technical overview.

\begin{lemma}[Distance-Preserving to Diameter-Preserving]\label{lem:dist-to-diam}
    Let $\textsc{LDD}$ be an LDD algorithm with slack $s$ that is $(\beta,\gamma)$-distance preserving and in a disconnected graph produces clusterings independently for each connected component.
    
    Then, there is a minimum growth factor $\alpha_0 = O(\gamma + \log n)$ such that for any $\alpha \geq \alpha_0$, the hierarchy $\{\mathcal{C}_\ell\}_{\ell \in [L]} = \textsc{LDDHierarchy}(G, \alpha, \textsc{LDD})$ is $(\epsilon,c)$-diameter preserving with high probability for $\epsilon = \frac{1}{L s + \beta}$ and $c = O(1)$.
\end{lemma}

\begin{proof}
    
    Assume without loss of generality that $\gamma = \Omega(\log n)$. Let $c_{succ}$ be a constant determining the desired success probability and let $c=9$. The precise value of $\alpha_0$ will be determined later, and will be chosen such that $\alpha_0 \ge \gamma$.
    
    Let $I_{\leq \ell}$ be the event that for every $\ell' \leq \ell$ and every cluster $C \in \mathcal{C}_{\ell'}$, $\diam(\widetilde{G}[C]) \leq c D_{\ell'}$. With this definition, $I_{\leq L}$ is equivalent to the hierarchy being $\epsilon$-diameter preserving. We prove by induction on increasing $\ell$ the following bound:
    \begin{equation*}
        \mathbb{P}(I_{\leq \ell}) \geq 1 - \ell n^{-c_{succ} - 1}.
    \end{equation*}
    
    For the base case, note that the event $I_{\leq 0}$ occurs with probability $1$, as the clusters of $\mathcal{C}_0$ are single vertices and thus have diameter $0$. The event $I_{\leq 1}$ occurs with high probability, as it is implied by the cluster diameter guarantees in $G$. 

    Suppose now that $\ell \geq 2$ and that the claim holds for strictly smaller $\ell$. We prove the claim for $\ell$. Fix the clusterings $\mathcal{C}_L, \dots, \mathcal{C}_\ell$, and suppose that the diameter guarantees in $G$ for these clusterings hold, so that $\diam(G[C]) \leq D_{\ell'}$ for every $\ell' \in \{L, \dots, \ell\}$ and $C \in \mathcal{C}_{\ell'}$, which occurs with high probability. Consider some cluster $C \in \mathcal{C}_\ell$ and two vertices $u, v \in C$. 

    Let $\widetilde{G}_{\ell-1}$ be $G$ with edge weights $\widetilde{w}_{\ell - 1}(e) \coloneqq w(e)+ \mathbb{I}[e \in \delta(\mathcal C_{\ell-1})]\cdot \epsilon D_{\ell-1}$.
    Since the LDD algorithm is $(\beta,\gamma)$-distance preserving and computes the clustering $\mathcal{C}_{\ell - 1}$ in the graph $G \setminus \delta (\mathcal{C}_\ell)$, there exists with probability at least $1 - n^{-c_{succ} - 4}$ a path $P$ in $\widetilde{G}_{\ell - 1}[C]$ of weight
    \begin{align*}
        \widetilde{w}_{\ell - 1}(P) &\leq (1  + \epsilon \beta) \dist_{G[C]}(u, v) + \gamma D_{\ell - 1}\\
            &\leq (1  + \epsilon \beta + \gamma / \alpha_0) D_\ell\\
            &\leq 3 D_\ell.
    \end{align*}
    Suppose this event occurs, and fix one such path. Let $C_1, \dots, C_k$ be the clusters of $\mathcal{C}_{\ell - 1}$ containing both endpoints of at least one edge on $P$, and let $P_i$ be the subset of edges of $P$ contained in this way inside $C_i$. Let $X_i$ be a random variable defined as
    \begin{equation*}
        X_i \coloneqq \min\{\widetilde{w}(P_i) - w(P_i), c D_{\ell - 1}\},
    \end{equation*}
    i.e. the minimum of the weight increase from the levels of the hierarchy below $\ell - 1$ to the edges $P_i$, or the diameter guarantee we obtain for clusters of the lower level. We have
    \begin{equation*}
        \mu \coloneqq \mathbb{E}\left[\sum_i X_i\right] \leq \epsilon s (\ell - 1) \sum_i w(P_i) \leq \sum_i w(P_i) \leq 3 D_\ell
    \end{equation*}
    as each of the $\ell - 1$ levels below increase in expectation the weight of any edge $e \in P$ by at most $\epsilon s \cdot w(e) \leq \frac{1}{L} w(e)$.
    Since each individual $X_i$ is additionally upper bounded by $c D_{\ell - 1}$, applying Chernoff (\Cref{thm:chernoffBound}) with $\delta \coloneqq \frac{3 D_\ell}{\mu} \geq 1$ gives
    \begin{align*}
        \mathbb{P}\left(\sum_i X_i \leq (1 + \delta) \mu\right)
        &\geq 1 - \exp\left(- \frac{\delta^2 \mu}{2 + \delta} \cdot \frac{1}{c D_{\ell - 1}}\right)\\
        &\geq 1 - \exp\left(- \frac{\delta \mu}{3} \cdot \frac{1}{c D_{\ell - 1}}\right)\\
        &\geq 1 - \exp\left(-\frac{D_\ell}{3 c D_{\ell - 1}}\right)\\
        &\geq 1 - \exp\left(- \frac{\alpha_0}{3c}\right)\\
        &\geq 1 - n^{- c_{succ} - 4},
    \end{align*}
    where we selected $\alpha_0 = \gamma + 3c(c_{succ} + 4) \log n = O(\gamma + \log n)$.
    Let $I_{uv}$ be the event that the diameter guarantees in $G$ hold, a desired path $P$ exists, and this upper bound on $\sum_i X_i$ holds. This event satisfies $\mathbb{P}(I_{uv}) \geq 1 - 3 n^{-c_{succ} - 4}$.
    
    Suppose that both $I_{uv}$ and $I_{\leq \ell-1}$ hold. We now show that this implies that $\dist_{\widetilde{G}[C]}(u, v) \leq c D_\ell$ for $c = 9$. It suffices to show that $\dist_{\widetilde G[C]}(u,v) \leq \widetilde w_{\ell-1}(P)+\sum_i X_i$. For this, we apply \Cref{clm:shortcutting} for initial graph $\widetilde G_{\ell-1}[C]$, perturbed graph $\widetilde G[C]$ and clustering $\mathcal C_{\ell-1}$. Since the weight of edges separated by $\mathcal C_{\ell-1}$ does not change, we may use any valid partition of the edges of $P$, and $M = 0$.

    Finally, extend the definition of $I_{uv}$ to vertices $u, v$ not sharing a cluster in $\mathcal{C}_\ell$ by defining $I_{uv}$ to always occur in this case. Now, if all events $I_{uv}$ for pairs $u, v \in V$ and the event $I_{\leq \ell - 1}$ occur, then the event $I_{\leq \ell}$ also occurs. We thus obtain by a union bound and the induction assumption that
    \begin{align*}
        \mathbb{P}\left(I_{\leq \ell}\right)    &\geq \mathbb{P}\left(I_{\leq \ell - 1}\right) - 3 n^2 \cdot n^{-c_{succ} - 4}\\
            &\geq 1 - (\ell - 1) n^{-c_{succ} - 1} - n^{-c_{succ} - 1}\\
            &= 1 - \ell \cdot n^{-c_{succ} - 1}
    \end{align*}
    as desired.
\end{proof}


\section{LSSTs From Diameter Preserving Hierarchies}

In this section, we analyze the perturbed-SSSP algorithm (\Cref{alg:low_stretch}), establishing \Cref{lem:hierarchy-to-lsst}.

\begin{algorithm}[h]
\caption{\textsc{PerturbedSSSP}($G$, $\alpha$, $\epsilon$, \textsc{LDD})}
\label{alg:low_stretch}
$r \gets \text{an arbitrary vertex}$\\
$L \gets \lceil \log_{\alpha} (nW) \rceil$\\
$\{\mathcal{C}_\ell\}_\ell \gets \textsc{LDDHierarchy}(G,\alpha, \textsc{LDD})$\\
\For{$e \in E$}{
    $\ell_e \gets \max\{\ell:e\in\delta(\mathcal{C}_\ell)\}$\\
    Sample $X_e \sim \mathrm{Geom}\left(\frac{3 \ln nW}{\epsilon D_{\ell_e}}\right)$\\
    $\widetilde{w}(e) \gets w(e) + \max(0, \epsilon D_{\ell_e} - X_e)$\\
}
$\widetilde{G} \gets (V,E,\widetilde{w})$\\
$\widetilde{T} \gets \text{shortest path tree rooted at $r$ in $\widetilde{G}$}$\\
\Return{$\widetilde{T}$}\\
\end{algorithm}

\begin{lemma}\label{lem:hierarchy-to-lsst}
    Suppose $\textsc{LDD}$ has slack $s$ and $\textsc{LDDHierarchy}(G,\alpha, \textsc{LDD})$ is $(\epsilon, c)$-diameter preserving with probability at least $1 - \frac{1}{n^3 W}$. Then, $\widetilde{T} = \textsc{PerturbedSSSP}(G, \alpha, \epsilon, \textsc{LDD})$ satisfies
    \begin{equation*}
        \mathbb{E}[\dist_{\widetilde{T}}(u, v)] \leq \dist_G(u, v) \cdot c\alpha L (\epsilon^{-1}+Ls) \cdot O(\log n),
    \end{equation*}
    where by $\widetilde T$ we denote the resulting spanning tree, with each edge $e\in E(\widetilde T)$ assigned weight $\widetilde w(e)$.
\end{lemma}

By linearity of expectation, it suffices to prove the claim for any pair $(u, v) \in E$ for which $\dist_G(u, v) = w(u, v)$. Fix one such pair. We define the following events:
\begin{itemize}
    \item Let $I_{\text{diam}}$ be the event that the hierarchy $\{\mathcal{C}_\ell\}_\ell$ is $(\epsilon, c)$-diameter preserving, i.e. that for every level $\ell \in [L]$ and every cluster $C \in \mathcal{C}_\ell$, $\diam(\widetilde{G}[C]) \leq c D_{\ell}$. By assumption, $\mathbb{P}(I_{\text{diam}}) \geq 1 - \frac{1}{n^3 W}$.
    
    \item Define $\overline{w}(e) = w(e) + \epsilon D_{\ell_e} - X_e$, and let $\overline{G} = (V, E, \overline{w})$ be $G$ with the (possibly negative) edge weights $\overline{w}$. Let $I_{\overline{w} = \widetilde{w}}$ be the event that $\overline{w}(e) = \widetilde{w}(e)$ for every edge $e$.
    
    A geometric tail bound gives
    \begin{equation*}
        \mathbb P[X_e > \epsilon D_{\ell_e}] \leq \left(1 - \frac{3 \ln nW}{\epsilon D_{\ell_e}}\right)^{\epsilon D_{\ell_e}} \leq \exp\left(-3 \ln nW\right) = (nW)^{-3}
    \end{equation*}
    for each $e \in E$. A union bound over at most $n^2$ edges yields $\mathbb P(I_{\overline{w}=\widetilde{w}}) \geq 1 - \frac{1}{nW}$.
    
    \item For $\ell \in [L]$, let $I_{\text{gap}, \ell}$ be an event defined as follows:
    \begin{itemize}
        \item If $u$ and $v$ appear in separate clusters in $\mathcal{C}_\ell$, $I_{\text{gap}, \ell}$ does not occur.
        \item Otherwise, $u, v \in C$ for some $C \in \mathcal{C}_\ell$. If also $r \in C$, then $I_{\text{gap}, \ell}$ occurs.
        \item Otherwise, $r \not\in C$. Let $P_1$ be a shortest edge-simple $r$-$u$ path in $\overline{G}$, and $P_2$ the shortest edge-simple $r$-$u$ path in $\overline{G}$ that does not share any edge in $\delta(C)$ with $P_1$. Then, $I_{\text{gap}, \ell}$ occurs when $\overline{w}(P_2) - \overline{w}(P_1) > 2 \widetilde{w}(u, v)$ regardless of the choice of $P_1$.
    \end{itemize}
\end{itemize}
We will bound the probability $I_{\mathrm{gap}, \ell}$ occurs later. First, the following lemma motivates the definitions of the individual events.

\begin{claim}\label{clm:good-events}
    Suppose that $I_{\mathrm{diam}}$, $I_{\overline{w} = \widetilde{w}}$, and $I_{\mathrm{gap}, \ell}$ all occur. Then, $\dist_{\widetilde{T}}(u, v) \leq 2c D_\ell$.
\end{claim}

\begin{proof}
    Let $C \in \mathcal{C}_\ell$ be the cluster containing both $u$ and $v$, which must exist as $I_{\mathrm{gap}, \ell}$ occurs. Let $P_1 = \widetilde{T}(r, u)$ be the tree path to $u$, and $P_2' = \widetilde{T}(r, v)$ be the tree path to $v$. First, note that it suffices to show that there is some vertex $x \in C$ that appears on both $P_1$ and $P_2'$, as then
    \begin{align*}
        \dist_{\widetilde{T}}(u, v)
            &\leq \dist_{\widetilde{T}}(u, x) + \dist_{\widetilde{T}}(x, v)\\
            &= \dist_{\widetilde{G}}(u, x) + \dist_{\widetilde{G}}(x, v)\\
            &\leq 2 \diam(\widetilde{G}[C])\\
            &\leq 2c D_\ell.
    \end{align*}
    If $r \in C$, then $r$ is one such $x$. If $u$ appears on $P_2'$, then $u$ is one such $x$. Otherwise, let $P_2$ be $P_2'$ extended with the edge $(v, u)$, and note that we have
    \begin{align*}
        \widetilde{w}(P_2)
            &= \widetilde{w}(u, v) + \widetilde{w}(P_2')\\
            &= \widetilde{w}(u, v) + \dist_{\widetilde{G}}(r, v)\\
            &\leq 2\widetilde{w}(u, v) + \dist_{\widetilde{G}}(r, u)\\
            &= 2\widetilde{w}(u, v) + \widetilde{w}(P_1).
    \end{align*}
    Since $P_1$ is a shortest $r$-$u$ path in $\widetilde{G}$ (and thus also edge-simple and a shortest edge-simple $r$-$u$ path in $\overline{G}$ as $I_{\overline{w} = \widetilde{w}}$ occurred), $P_2$ is an edge-simple $r$-$u$ path of length at most $2 \widetilde{w}(u, v)$ more than $P_1$ in $\widetilde{G}$ (and thus also in $\overline{G}$), and $I_{\mathrm{gap}, \ell}$ occurred, there must be some edge $e \in \delta(C)$ that appears on both $P_1$ and $P_2$. That edge cannot be $(u, v)$ since both $u, v \in C$, thus $e$ also appears on $P_2'$. We can select $x$ to be the endpoint of $e$ inside $C$.
\end{proof}

To bound the probability $I_{\mathrm{gap}, \ell}$ occurs, we will use the following Lemma:

\begin{lemma}\label{lem:gap}
    Let $G = (V, E, w)$ be a graph with possibly negative edge weights, $s, t \in V$ two distinct vertices, and $E' \subseteq E$ any subset of edges such that there exists a shortest edge-simple $s$-$t$ path using at least one edge from $E'$.
    
    For every $e \in E'$, let $X_e \sim \mathrm{Geom}(p_e)$ be independent geometrically distributed random variables with individual parameters $p_e \in (0, p_{\max}]$ for $p_{\mathrm{max}} < 1$. Let $\overline{G} = (V, E, \overline{w})$ be $G$ with edge weights $\overline{w}(e) = w(e) - X_e$ for $e \in E'$ and $\overline{w}(e) = w(e)$ for all other edges.

    Let $P_1$ be a shortest edge-simple $s$-$t$ path in $\overline{G}$, and $P_2$ a shortest edge-simple $s$-$t$ path in $\overline{G} \setminus (E' \cap P_1)$, i.e. that does not use any of the edges in $E'$ that $P_1$ does. Then, for every $k \in \mathbb{Z}_{\geq 0}$,
    \begin{equation*}
        \mathbb{P}\left(\overline{w}(P_2) - \overline{w}(P_1) \geq k\right)
        \geq \left(1 - p_{\max}\right)^k.
    \end{equation*}
\end{lemma}

\begin{proof}
    Consider the following process that samples values for the variables $X_e$:
    \begin{itemize}
        \item Initially, let $X_e = 0$ for every $e \in E'$, and let $P$ be a shortest edge-simple $s$-$t$ path in $G$ on which at least one edge in $E'$ appears.
        \item Repeat the following until $E' = \emptyset$:
        \begin{itemize}
            \item Select an edge $e$ from $E'$. If there is an edge $e \in E' \setminus P$, select one such edge. Otherwise, select an arbitrary edge in $E'$.
            \item With probability $p_e$, remove $e$ from $E'$. Otherwise, let $X_e \gets X_e + 1$, then assign $P$ to be an arbitrary shortest edge-simple $s$-$t$ path in $\overline{G}$ on which at least one edge in $E'$ appears.   
        \end{itemize}
    \end{itemize}
    This process terminates in finite time, as $p_e > 0$ for every edge $e \in E'$. The final joint distribution of $X_e$ is the desired distribution: one can easily verify that $\mathbb{P}(X = X') = \prod_{e \in E'} (1 - p_e)^{X_e'} p_e$ for any fixed vector $X' \in \mathbb{Z}_{\geq 0}^{E'}$ of target values.

    We now first prove that initially and after each step of the process, $P$ is a shortest edge-simple $s$-$t$ path in $\overline{G}$. This holds initially by the requirement placed on $E'$. In each step, if we remove the selected edge from $E'$, no weights change, so $P$ must still be a shortest path. If we decrease the weight of the selected edge $e$, and this makes $P$ no longer be a shortest edge-simple $s$-$t$ path, then the new shortest path must use the selected edge $e$, which remains in $E'$.

    Consider now the first point in the process where $|E'| = 1$, and let $e$ be the sole remaining edge and $d = \overline{w}(P)$ the weight of $P$ in the current graph at that point. Since $P$ at that point is a shortest edge-simple $s$-$t$ path, any edge-simple $s$-$t$ path on which $e$ does not appear must have weight at least $d$ in the final graph regardless of how the process continues.
    
    For any $k \in \mathbb{Z}_{\geq 0}$, let $I_k$ be the event that $X_e$ gets incremented at least $k$ times from this point onwards. We have $\mathbb{P}(I_k) = (1 - p_e)^k \geq (1 - p_{\max})^k$. Conditioning on $I_k$, we have $\overline{w}(P) \leq d - k$ in the final graph. Since any path not using $e$ must have weight at least $d$, $P_1$ must also use $e$ and have weight at most $\overline{w}(P_1) \leq d - k$. Since $P_1$ uses $e$, $P_2$ cannot use $e$, and must thus have weight at least $d$. Thus, $I_k$ implies $\overline{w}(P_2) - \overline{w}(P_1) \geq k$.
\end{proof}

\begin{claim}\label{clm:gap-prob}
    For each $\ell \in [L]$, 
    \begin{equation*}
        \mathbb{P}(I_{\mathrm{gap},\ell}) \ge 1-\frac{w(u,v)}{D_\ell} (\epsilon^{-1} + Ls) O(\log n).
    \end{equation*}
\end{claim}
\begin{proof}
    To make conditionals more readable, we write $\mathcal{R}_1$ to denote the levels $\mathcal{C}_L, \dots, \mathcal{C}_\ell$ of the hierarchy, $\mathcal{R}_2$ to denote the remaining levels, $\mathcal{R}_3$ to denote the values $X_e$ for edges $e \not\in \delta(C)$ where $C$ is the level-$\ell$ cluster containing $u$, and $\mathcal{R}_4$ the values $X_e$ for all remaining edges. We denote realizations of these by $\mathcal{R}_1', \dots, \mathcal{R}'_4$.
    
    Since the LDD algorithm has slack $s$, there exists a cluster $C \in \mathcal{C}_\ell$ for which $u, v \in C$ with probability at least
    \begin{equation*}
        1 - \sum_{\ell' = \ell}^{L} s \cdot \frac{w(u, v)}{D_{\ell'}} \geq 1 - 2 s \frac{w(u, v)}{D_\ell}.
    \end{equation*}
    Suppose this occurs, and fix levels $\mathcal{C}_L, \dots, \mathcal{C}_\ell$ of the hierarchy for which $u, v \in C \in \mathcal{C}_\ell$. The expected $\widetilde{w}(u, v)$ satisfies
    \begin{equation*}
        \mathbb{E}[\widetilde{w}(u, v) \mid \mathcal{R}_1 = \mathcal{R}'_1] \leq w(u, v) + \sum_{\ell' = 0}^{\ell - 1} \epsilon D_{\ell'} \cdot s \frac{w(u, v)}{D_{\ell'}} \leq w(u, v) \left(1 + \epsilon L s\right).
    \end{equation*}
    Next, fix the levels $\mathcal{C}_{\ell - 1}, \dots, \mathcal{C}_0$ and values $X_e$ for $e \not\in \delta(C)$ arbitrarily. Let $p_{\mathrm{max}} \coloneqq \frac{3 \ln {nW}}{\epsilon D_{\ell}}$. If $p_{\mathrm{max}} \geq 1$, then the claim holds trivially, and we can thus assume $p_{\mathrm{max}} < 1$. We now apply \Cref{lem:gap} to $E' \coloneqq \delta(C)$, $k = \lfloor 2\widetilde{w}(u, v) \rfloor + 1$, and $G' \coloneqq (V, E, w')$ with
    \begin{equation*}
        w'(e) \coloneqq w(e) + \begin{cases}
            \epsilon D_{\ell_e} - X_e   &e \not\in \delta(C)\\
            \epsilon D_{\ell_e}         &e \in \delta(C)
        \end{cases}.
    \end{equation*}
    This gives
    \begin{equation*}
        \mathbb{P}\left(I_{\text{gap}, \ell} \mid \mathcal{R}_{1, 2, 3} = \mathcal{R}_{1, 2, 3}'\right) \geq (1 - p_{\mathrm{max}})^{k} \geq 1 - k \cdot p_{\mathrm{max}} \geq 1 - \frac{\widetilde{w}(u, v)}{D_\ell} \epsilon^{-1} \cdot 10 \ln nW.
    \end{equation*}
    We can thus apply the law of total probability to obtain
    \begin{align*}
        \mathbb{P}(I_{\text{gap}, \ell} \mid \mathcal{R}_1 = \mathcal{R}'_1)
            &\geq \sum_{\mathcal{R}'_{2, 3}} \left(1 - \frac{\widetilde{w}(u, v)}{D_\ell} \epsilon^{-1} \cdot 10 \ln nW\right) \mathbb{P}(\mathcal{R}_{2, 3} = \mathcal{R}'_{2, 3}\mid \mathcal{R}_1 = \mathcal{R}'_1)\\
            &= 1 - \frac{10 \ln nW}{D_\ell} \epsilon^{-1} \sum_{\mathcal{R}'_{2, 3}} \widetilde{w}(u, v) \cdot\mathbb{P}(\mathcal{R}_{2, 3} = \mathcal{R}'_{2, 3}\mid \mathcal{R}_1 = \mathcal{R}'_1)\\
            &\geq 1 - \frac{10 \ln nW}{D_\ell} \epsilon^{-1} \left(w(u, v) (1 + \epsilon L s)\right)\\
            &= 1 - \frac{w(u, v)}{D_\ell} (\epsilon^{-1} + Ls) 10 \ln nW.
    \end{align*}
    Thus,
    \begin{align*}
        \mathbb{P}(I_{\text{gap}, \ell})
            &\geq 1 - 2 s \frac{w(u, v)}{D_\ell} - \frac{w(u, v)}{D_\ell} (\epsilon^{-1} + Ls) 10 \ln nW\\
            &= 1 - \frac{w(u, v)}{D_\ell} (\epsilon^{-1} + Ls) O(\log n)
    \end{align*}
    as desired.
\end{proof}

We now have everything required to prove \Cref{lem:hierarchy-to-lsst}. A union bound over all bad events combined with \Cref{clm:gap-prob} yields
\begin{equation*}
\mathbb{P}(I_{\mathrm{diam}} \land I_{\overline{w} = \widetilde{w}} \land I_{\mathrm{gap}, \ell}) \ge 1-\frac{w(u,v)}{D_\ell}(\epsilon^{-1}+Ls)O(\log n)-\frac{2}{nW}.    
\end{equation*}
Since $D_\ell = O(\alpha nW)$, the second term dominates. Combining this with \Cref{clm:good-events} yields
\begin{equation*}
    \mathbb{P}(\dist_{\widetilde{T}}(u, v) \leq 2c D_\ell) \ge 1-\frac{w(u,v)}{D_\ell}(\epsilon^{-1}+Ls)O(\log n).
\end{equation*}
Finally, if $I_{\text{diam}}$ occurs, $\dist_{\widetilde{T}}(u,v)\leq 2c D_L$, and $\dist_{\widetilde{T}}(u,v)\leq 2n D_L = O(\alpha n^2 W)$ holds always as the maximum weight of an edge in $\widetilde{T}$ is $W + \epsilon D_L$, thus the contribution to $\mathbb{E}[\dist_{\widetilde{T}}(u,v)]$ from $I_{\text{diam}}$ not occurring is at most $O\left(\frac{\alpha n^2 W}{n^3 W}\right) = O(1)$. Therefore, as $w(u, v) \geq 1$,
\begin{align*}
\mathbb{E}[\dist_{\widetilde T}(u,v)]
&\leq O(c) + \sum_{\ell=0}^{L-1} 2cD_{\ell+1} \mathbb P(2cD_\ell < \dist_{\widetilde T}(u,v) \le 2cD_{\ell+1})\\
&\leq w(u,v) \cdot c(\epsilon^{-1}+Ls)O(\log n) \sum_{\ell=0}^{L-1}\frac{D_{\ell+1}}{D_\ell}\\
&= w(u,v)\cdot c \alpha L (\epsilon^{-1}+Ls)O(\log n).
\end{align*}
This completes the proof of \Cref{lem:hierarchy-to-lsst}.

\section{Proof of the Main Theorems}

\LowStretchTree*

\begin{proof}[Proof for GEOM]
    We use GEOM, given in \Cref{alg:geom}, as our LDD algorithm. By \Cref{lem:geom-distance-preserving} it is $(\beta, \gamma)$-distance preserving for $\beta, \gamma = O(\log n)$. By \Cref{lem:dist-to-diam}, we can choose $\alpha = \Theta(\log n)$ such that $\textsc{LDDHierarchy}(G, \alpha, \textsc{GEOM})$ is $(\epsilon,c)$-diameter preserving with $\epsilon = \frac{1}{L s+\beta}$ and $c = O(1)$ with a sufficiently high probability. By \Cref{lem:hierarchy-to-lsst}, \Cref{alg:low_stretch} returns a spanning tree that satisfies 
    \begin{align*}
        \mathbb{E}[\dist_{\widetilde{T}}(u, v)]
            &\leq \dist_G(u, v) \cdot c \alpha L (\epsilon^{-1}+Ls)O(\log n)\\
            &= \dist_G(u, v) \cdot L^2 \cdot O(\log^3 n)\\
            &= \dist_G(u, v) \cdot O\left(\frac{\log^5 n}{(\log \log n)^2}\right).
    \end{align*}
    for any $u,v \in V$, as $L \coloneqq \lceil \log_\alpha(nW) \rceil = \Theta(\frac{\log n}{\log \log n})$ and $\epsilon = \Theta(1 / L \log(n))$. Since $\widetilde w(e) \ge w(e)$ for all $e$, we have $\dist_{T}(u, v) \le \dist_{\widetilde{T}}(u, v)$, which completes the proof of the desired stretch.

    The $\widetilde{O}(m)$-time complexity follows immediately from the existence of $\widetilde{O}(m)$-time implementations of GEOM and Dijkstra's algorithm.
\end{proof}

\begin{proof}[Proof for MPX]
    We use the same parameters as in the proof for GEOM, except that $\beta = O(\log^2 n)$. Substituting yields a slightly worse but sufficient guarantee of $O(\frac{\log^5 n}{\log \log n}) = O(\log^5 n)$. MPX can be implemented in $\widetilde{O}(m)$ time \cite{miller2013mpxoriginal}.
\end{proof}

\mainthm*

\begin{proof}
    Fix an arbitrary root vertex $r\in V$. \Cref{alg:low_stretch} using either GEOM or MPX and $\epsilon = \frac{\varepsilon}{2Ls+\beta}$ constructs perturbations $\Delta$ and the shortest-path tree $\widetilde{T}_r$ in the perturbed graph $\widetilde{G}$. These perturbations are nonnegative and satisfy
    \begin{equation*}
        \mathbb{E}[\Delta(e)] \leq \sum_{\ell = 0}^{L} \left(s \cdot \frac{w(e)}{D_\ell}\right) \cdot \epsilon D_\ell = w(e) \sum_{\ell = 0}^{L} \frac{\varepsilon s}{2Ls + \beta} \leq \varepsilon w(e).
    \end{equation*}
    Since $\widetilde{w}(e) \ge w(e)$ for all $e$, for every edge $(u,v)\in E$,
    \begin{equation*}
        \mathbb{E}\left[\frac{\dist_{\widetilde T_r}(u,v)}{\dist_{\widetilde G}(u,v)}\right]
        \leq \mathbb{E}\left[\frac{\dist_{\widetilde T_r}(u,v)}{\dist_{G}(u,v)}\right]
        \leq \widetilde{O}\left(\frac{1}{\varepsilon}\right),
    \end{equation*}
    where the last step follows from the proof of \Cref{thm:LowStretchTree}.
    \end{proof}

\section*{Acknowledgments}

The authors would like to thank Jakob Nogler, Gary Miller, and Goran Zuzic for helpful discussions.

\newpage
\bibliographystyle{alpha}
\bibliography{refs}

\pagebreak
\appendix

\section{Deferred Proofs}\label{sec:deferred}

\geom*

\begin{proof}
    For the diameter, consider a fixed sample $R$. By the tail bound of the geometric distribution,
    \begin{equation*}
        \mathbb{P}\left(R \geq \frac{D}{2}\right) \leq \left(1 - \frac{c_{\text{geom}} \log n}{D}\right)^{\frac{D}{2}}
            \leq \exp\left(-\frac{c_{\text{geom}}}{2} \log n\right)
            = n^{-\Theta(c_{\text{geom}})}
    \end{equation*}
    since each carved ball contains at least one vertex, at most $n$ samples are made, thus by the union bound, all samples are at most $\frac{D}{2}$ with probability $1 - n^{-\Theta(c_{\text{geom}})}$. If this occurs, each cluster has diameter at most $D$.

    For the cutting probability, fix some edge $e = (u, v)$ and center $x$. Without loss of generality, assume that $x$ is closer to $u$ than $v$ in $G[U]$, and let $R_0 \coloneqq \dist_{G[U]}(x, u)$ and $R_1 \coloneqq \dist_{G[U]}(x, v)$. Then, $e \in \delta(B(x, R))$ if and only if $R \in [R_0, R_1)$. Thus, conditioning on $R < R_0$, the probability that $B(x, R)$ cuts $e$ is zero. Conditioning on $R \geq R_0$, by the memorylessness of the geometric distribution,
    \begin{align*}
        \mathbb{P}(R \geq R_1 \mid R \geq R_0)
            &= (1 - p)^{R_1 - R_0}\\
            &= \left(1 - \frac{c_{\text{geom}} \log n}{D}\right)^{R_1 - R_0}\\
            &\geq 1 - \frac{R_1 - R_0}{D} c_{\text{geom}} \log n\\
            &\geq 1 - s \cdot \frac{w(e)}{D}.
    \end{align*}
    for $s \coloneqq c_{\text{geom}} \log n$, where in the last step we use $R_1 - R_0 \leq w(u, v)$ by the triangle inequality. Thus, the probability that $e$ gets cut by $B(x, R)$ given $R \geq R_0$ is at most the desired cutting probability. If this does not occur, then $e \in B(x, R)$, and cannot be cut in the future.
\end{proof}

\randpermprefmins*

\begin{proof}
    Consider fixing the permutation in the order $\pi_n, \pi_{n - 1}, \dots, \pi_1$, selecting each $\pi_i$ uniformly randomly from $[n] \setminus \{\pi_{i + 1}, \dots, \pi_n\}$. Let $X_i$ be an indicator that $\pi_i$ is a prefix minimum. Then,
    \begin{equation*}
        \mathbb{P}(X_i \mid \pi_{i + 1}, \dots, \pi_{n}) = \frac{1}{i}
    \end{equation*}
    since $X_i = 1$ if and only if $\pi_i = \min \left([n] \setminus \{\pi_{i + 1}, \dots, \pi_n\}\right)$ and $\pi_i$ is selected uniformly at random from this set of size $i$. Since this probability does not depend on $\pi_{i + 1}, \dots, \pi_n$, the variables $X_i$ are independent. Additionally,
    \begin{equation*}
        \mu \coloneqq \mathbb{E}\left[\sum_{i = 1}^{n} X_i\right] = \sum_{i = 1}^{n} \frac{1}{i} = H_n = \Theta(\log n)
    \end{equation*}
    where $H_n$ is the $n$\textsuperscript{th} harmonic number. Thus, by Chernoff (\Cref{thm:chernoffBound}), for $\delta \geq 1$, 
    \begin{equation*}
        \mathbb{P}\left(\sum_i X_i \geq (1 + \delta) \mu\right) \leq \exp\left(-\frac{\delta^2 \mu}{2 + \delta}\right) \leq \exp\left(-\mu \frac{\delta}{3}\right) = n^{-\Theta(\delta)}.
    \end{equation*}
    Selecting $\delta$ to be a desiredly large constant finishes the proof.
\end{proof}

\shortcutting*

\begin{proof}
First, pay $w(P)$ to account for the original cost of every edge. For each cluster $C_i \in \mathcal C$, choose the cheaper of the following two options:

\begin{itemize}
\item We pay $\Delta(P_i)$. If we do this, we may use all edges in $P_i$.
\item We pay $D+2M$. If we do this, we may shortcut through $C_i$ once. In particular, we may move between any two vertices of $C_i$ at a cost of at most $D$. We may additionally use the first and last edges on $P$ intersecting $C_i$, each of which increased in weight by at most $M$.
\end{itemize}

We construct a new $u$--$v$ path that respects these constraints. Let $x$ denote the current vertex, with initially $x \gets u$.

We repeat the following: Let $(x,y)$ be the edge of $P$ incident to $x$ such that $y$ is
closer to $v$ along $P$, and suppose $(x,y) \in P_i$. If we chose the first option for $C_i$, we traverse the edge $(x,y)$ and set $x \gets y$. Otherwise, among all endpoints of edges in $P_i$, let $z$ be the one closest to $v$ along $P$. We shortcut through $C_i$ from $x$ to $z$ and set $x \gets z$.

In each iteration, the number of hops remaining from the current vertex to $v$ along $P$ decreases by at least one. Consequently, every edge assigned to a cluster of the first type is traversed at most once, while each cluster of the second type is shortcut through at most once.
\end{proof}

\end{document}